\documentclass[sigconf, nonacm]{acmart}

\usepackage{booktabs}
\usepackage{balance}
\usepackage{marvosym}
\usepackage{enumitem}
\usepackage{multirow}
\usepackage{bm}
\usepackage[ruled,linesnumbered,vlined]{algorithm2e}
\usepackage{amsthm} % for proof without number
\usepackage{hyperref}

\usepackage[table]{xcolor}

\definecolor{ParaMarkColor}{RGB}{0,0,0}

\newcommand{\paramark}[1]{%
  \textcolor{ParaMarkColor}{#1}%
}

\definecolor{bestgray}{gray}{0.85}
\definecolor{secondgray}{gray}{0.93}

\newcommand{\bestcell}[1]{\cellcolor{bestgray}{#1}}
\newcommand{\secondcell}[1]{\cellcolor{secondgray}{#1}}

\newcommand{\gbest}[1]{\textbf{\underline{#1}}}
\newcommand{\gsecond}[1]{\underline{#1}}

\newtheorem{myDef}{\noindent \textbf{Definition}}

\newcommand\vldbdoi{XX.XX/XXX.XX}
\newcommand\vldbpages{XXX-XXX}
\newcommand\vldbvolume{14}
\newcommand\vldbissue{1}
\newcommand\vldbyear{2020}
\newcommand\vldbauthors{\authors}
\newcommand\vldbtitle{\shorttitle} 
\newcommand\vldbavailabilityurl{https://github.com/LIWEIDENG0830/LB-TrajRep}
\newcommand\vldbpagestyle{plain} 

\begin{document}
% \title{Revisiting Trajectory Similarity Learning via Lower-Bound Representations}
\title{Using Lower-Bound Representations for Trajectory Similarity Learning}

%%
%% The "author" command and its associated commands are used to define the authors and their affiliations.
\author{Liwei Deng}
\affiliation{%
  \institution{Aalborg University}
  % \streetaddress{P.O. Box 1212}
  % \city{Aalborg}
  % \state{Denmark}
  % \postcode{43017-6221}
}
\email{lide@cs.aau.dk}

\author{Haotian Meng, Yupu Zhang}
% \orcid{0000-0002-1825-0097}
\affiliation{%
  \institution{University of Electronic Science and Technology of China}
  % \streetaddress{1 Th{\o}rv{\"a}ld Circle}
  % \city{Hekla}
  % \country{China}
  % \city{Shenzhen}
  % \country{China}
}
\email{menghaotian, zhangyupu@std.uestc.edu.cn}

% \author{Yupu Zhang}
% % \orcid{0000-0002-1825-0097}
% \affiliation{%
%   \institution{University of Electronic Science and Technology of China}
%   % \streetaddress{1 Th{\o}rv{\"a}ld Circle}
%   % \city{Hekla}
%   % \country{China}
%   % \city{Chengdu}
%   % \country{China}
% }
% \email{zhangyupu@std.uestc.edu.cn}

\author{Yan Zhao}
% \orcid{0000-0002-1825-0097}
\affiliation{%
  \institution{University of Electronic Science and Technology of China\Letter}
  % \streetaddress{1 Th{\o}rv{\"a}ld Circle}
  % \city{Hekla}
  % \country{China}
  % \city{Chengdu}
  % \country{China}
}
\email{zhaoyan@uestc.edu.cn}

\author{Torben Bach Pedersen}
\affiliation{%
  \institution{Aalborg University}
  % \city{Aalborg}
  % \country{Denmark}
}
\email{tbp@cs.aau.dk}

\author{Kai Zheng}
\affiliation{%
  \institution{University of Electronic Science and Technology of China\Letter} % \Letter
  % \streetaddress{P.O. Box 1212}
  % \city{Dublin}
  % \state{Ireland}
  % \postcode{43017-6221}
  % \city{Chengdu}
  % \country{China}
}
\email{zhengkai@uestc.edu.cn}
\authornote{\Letter \ The corresponding author.}

\author{Christian S. Jensen}
\affiliation{%
  \institution{Aalborg University}
  % \city{Aalborg}
  % \country{Denmark}
}
\email{csj@cs.aau.dk}

%%
%% The abstract is a short summary of the work to be presented in the
%% article.
\begin{abstract}

Trajectory similarity learning is fundamental to efficient trajectory retrieval under complex distance measures. Existing learning-based methods typically rely on embeddings trained to approximate trajectory distances or rankings, but they often lack guarantees with respect to the original distances, exhibit unstable performance across distance measures, and incur substantial training costs.
% In this work, we revisit trajectory similarity learning from a lower-bound representation perspective. We propose a unified pivot-based framework that represents trajectories by their distances to a set of pivot points, inducing admissible and interpretable lower bounds for multiple classical trajectory distances, including Dynamic Time Warping (DTW), Hausdorff distance, and Discrete Fr\'echet Distance (DFD). The proposed framework naturally supports both metric and non-metric distances and is fully compatible with standard single-vector retrieval pipelines. To improve representation quality, we further develop two data-driven pivot selection strategies that explicitly optimize lower-bound tightness and prioritize hard near-neighbor trajectory pairs.
We revisit trajectory similarity learning from a lower-bound representation perspective and \paramark{propose LB-TrajRep, a unified lower-bound representation framework independent of deep neural embeddings. This framework constructs single-vector representations from a set of lower-bound components}, enabling admissible and interpretable lower bounds for multiple classical trajectory distances, including Dynamic Time Warping (DTW), Hausdorff distance, and Discrete Fr\'echet Distance (DFD).
% , through appropriate component instantiations. 
Within this framework, we instantiate point-pivot components, which naturally support both metric and non-metric distances and remain compatible with standard vector-based retrieval pipelines. To improve ranking quality, we develop two data-driven pivot selection strategies that explicitly optimize lower-bound tightness and prioritize hard near-neighbor trajectory pairs, respectively.
Extensive experiments on real-world trajectory datasets show that the proposed lower-bound representations are able to consistently outperform state-of-the-art neural trajectory embeddings across diverse distance measures, improving top-$k$ ranking accuracy by up to 20\%--60\% on the Hausdorff distance and DFD and by 15\%--40\% on DTW.
% , while achieving substantially improved ranking stability.

\end{abstract}

\maketitle

%%% do not modify the following VLDB block %%
%%% VLDB block start %%%
\pagestyle{\vldbpagestyle}
\begingroup\small\noindent\raggedright\textbf{PVLDB Reference Format:}\\
\vldbauthors. \vldbtitle. PVLDB, \vldbvolume(\vldbissue): \vldbpages, \vldbyear. \href{https://doi.org/\vldbdoi}{doi:\vldbdoi}
\endgroup
\begingroup
\renewcommand\thefootnote{}\footnote{\noindent
This work is licensed under the Creative Commons BY-NC-ND 4.0 International License. Visit \url{https://creativecommons.org/licenses/by-nc-nd/4.0/} to view a copy of this license. For any use beyond those covered by this license, obtain permission by emailing \href{mailto:info@vldb.org}{info@vldb.org}. Copyright is held by the owner/author(s). Publication rights licensed to the VLDB Endowment. \\
\raggedright Proceedings of the VLDB Endowment, Vol. \vldbvolume, No. \vldbissue\ %
ISSN 2150-8097. \\
\href{https://doi.org/\vldbdoi}{doi:\vldbdoi} \\
}\addtocounter{footnote}{-1}\endgroup
%%% VLDB block end %%%

%%% do not modify the following VLDB block %%
%%% VLDB block start %%%
\ifdefempty{\vldbavailabilityurl}{}{
\vspace{.3cm}
\begingroup\small\noindent\raggedright\textbf{PVLDB Artifact Availability:}\\
The source code, data, and other artifacts have been made available at \url{\vldbavailabilityurl}.
\endgroup
}
%%% VLDB block end %%%

\section{Introduction}

Trajectory similarity computation~\cite{yao2022trajgat,fang2022spatio,cao2021accurate,zhang2021trajectory,chang2023contrastive,yang2021t3s,yang2022tmn,hu2023spatio,zhou2023grlstm,deng2024learning,wang2019effective,cao2024hypergraph,chang2025general} is a fundamental problem in spatio-temporal data analysis, with applications ranging from trajectory retrieval~\cite{chang2024revisiting,yang2024simformer} and clustering~\cite{atev2010clustering,gonzalez1985clustering} to traffic analysis and anomaly detection~\cite{zhang2020continuous}. 
% Given a collection of trajectories and a query trajectory, the task is to retrieve trajectories that are most similar under a predefined distance measure, 
Given two trajectories, the task is to determine their similarity according to such as Dynamic Time Warping (DTW)~\cite{yi1998efficient}, Hausdorff distance~\cite{atev2010clustering}, or Discrete Fréchet Distance (DFD)~\cite{alt1995computing}. While offering high utility, these classical distance measures are computationally expensive, typically having quadratic time complexity in the trajectory length, which is prohibitive for large-scale datasets~\cite{yao2019computing}.

To lower computational costs, trajectory similarity learning has emerged~\cite{hu2025simrn}. The core idea is to represent trajectories as a vectors such that the similarity between trajectories can be approximated by distance computations on fixed-length vectors~\cite{yao2019computing}. A large body of recent studies propose neural network based methods to learn trajectory embeddings. Representative approaches include RNN-based methods that encode trajectories as point sequences, e.g., NeuTraj~\cite{yao2019computing} and its variants, Transformer-based methods that capture long-range dependencies through self-attention, e.g., TrajGAT~\cite{yao2022trajgat} and SIMformer~\cite{yang2024simformer}, and CNN-based methods that emphasize local spatial patterns, e.g., TrjSR~\cite{cao2021accurate} and ConvTraj~\cite{chang2024revisiting}. More recent studies incorporate ranking-aware losses~\cite{deng2024learning,chang2025k}, graph structures~\cite{yao2022trajgat,han2021graph}, or multi-granularity encoders~\cite{yao2022trajgat} to improve approximation accuracy. Despite their architectural diversity, existing learning-based methods share a common goal: approximating complex trajectory distances or their induced similarity rankings using embeddings encoded by neural networks~\cite{hu2025simrn}. However, empirical evidence from our study suggests that this general approach faces several fundamental limitations.

\textbf{C1. Learned trajectory embeddings generally provide no guarantees with respect to the original trajectory distances and are difficult to interpret.} 
Most existing learning-based methods encode trajectories into latent embedding spaces where similarity is measured by a simple vector distance. However, such embedding distances typically provide no explicit guarantees with respect to the original trajectory distance, making it difficult to reason about their semantic consistency. This limitation is particularly relevant for non-metric or alignment-based distances such as DTW~\cite{berndt1994using}, whose values are sensitive to local variations and alignment choices~\cite{cuturi2011fast,su2020survey}. Moreover, neural trajectory embeddings are often opaque, offering little insight into why two trajectories are considered similar or dissimilar, which complicates failure analysis and systematic improvement.

% Most existing methods are trained using regression-style losses~\cite{yao2019computing,yang2024simformer} or ranking-based objectives~\cite{chang2024revisiting,deng2024learning} over sampled trajectory pairs. 
% As a result, the learned embedding distance may approximate the ground-truth distance for some trajectory pairs while failing to preserve the relative similarity ordering for others. This issue is particularly pronounced for non-metric or alignment-based distances such as DTW~\cite{berndt1994using}, where small local variations or alignment choices can significantly affect the true distance~\cite{cuturi2011fast,su2020survey}. 
% As a result, the learned embedding distance can behave inconsistently outside the sampled supervision and provides no explicit guarantee with respect to the original distance measure. This limitation is particularly relevant for non-metric or alignment-based distances such as DTW~\cite{berndt1994using}, whose values are sensitive to local variations and alignment choices~\cite{cuturi2011fast,su2020survey}. Moreover, neural trajectory embeddings are typically opaque~\cite{lipton2018mythos,doshi2017towards}, offering little insight into why two trajectories are considered similar or dissimilar, which makes failure cases hard to analyze and systematic improvements difficult. In practice, this often leads to unstable ranking performance across datasets and distance measures.

\textbf{C2. Many learning-based approaches exhibit a mismatch between their modeling assumptions and the retrieval setting.} 
Trajectory similarity learning is evaluated predominantly according to top-$k$ ranking quality. Yet, most learning-based approaches are trained to minimize distance regression errors~\cite{yao2019computing,yang2024simformer} or pairwise surrogate losses over sampled trajectory pairs~\cite{chang2024revisiting,deng2024learning}. Such objectives emphasize average approximation accuracy rather than explicitly prioritizing ranking-critical near-neighbor trajectories that affect retrieval performance the most. As a result, improvements in training loss do not necessarily translate into better ranking quality in similarity search scenarios.

\textbf{C3. Existing trajectory similarity learning methods are often tightly coupled to specific distance measures and incur substantial training costs.} 
% Trajectory distances such as DTW~\cite{yi1998efficient}, Hausdorff distance~\cite{huttenlocher2002comparing,atev2010clustering}, and DFD~\cite{alt1995computing} capture similarity from fundamentally different perspectives, yet most learning-based approaches require distance-specific supervision~\cite{yang2024simformer,zhang2021trajectory}, or architectural design~\cite{deng2024learning,yao2019computing}. This fragmentation makes it difficult to develop unified representations that generalize across multiple trajectory distances, and it often forces practitioners to redesign objectives, re-tune hyperparameters, and retrain models whenever the target distance or data characteristics change.  Combined with the high cost of training deep neural models, this leads to significant overhead in both data preparation and model training, limiting the practicality and scalability of learning-based solutions~\cite{chang2024trajectory,hu2023spatio}.
Distances such as Hausdorff~\cite{huttenlocher2002comparing,atev2010clustering}, DFD~\cite{alt1995computing}, and DTW~\cite{yi1998efficient} capture similarity from fundamentally different perspectives. Further, existing learning-based methods often rely on distance-specific supervision~\cite{yang2024simformer,zhang2021trajectory} or architectural designs. Consequently, changing the target distance typically requires redesigning objectives, retuning  of hyperparameters, and retraining of models. Combined with the high computational cost of training deep neural networks, this tight coupling to specific distance measures limits representation reuse across distances and reduces the practicality of learning-based solutions in real-world systems.

In this paper, we study trajectory similarity learning adopting a new perspective. Instead of learning embeddings using neural networks that approximate trajectory distances directly, we propose to learn representations whose induced distances form tight, admissible lower bounds of classical trajectory distances (\textbf{C1}). This new approach aligns the representation learning objective more closely with the retrieval task, where preserving similarity ordering is more critical than minimizing absolute distance errors, while still respecting the structural properties of the original distance measures. 
Specifically, we propose a pivot-based lower-bound representation framework, where each trajectory is represented by its distances to a set of pivot points in the spatial domain, computed via simple point-to-trajectory distance primitives (\textbf{C2}). By defining pivots as points rather than trajectories, the framework naturally accommodates both metric and non-metric trajectory distances (\textbf{C3}). With this representation, the similarity between trajectories can be computed directly on their representations, yielding valid lower bounds for Hausdorff distance and DFD, as well as for DTW via one-dimensional projected sequences. Importantly, this design does not modify the original distance definitions and remains fully compatible with standard retrieval pipelines based on vector representations.

% A central challenge in pivot-based representations lies in selecting effective pivots. Poorly chosen pivots may lead to loose bounds and weak ranking performance, negating the benefits of the representation. To address this issue, we propose two data-driven pivot selection algorithms tailored for trajectory similarity learning. The first, \textbf{tightness-aware pivot selection}, explicitly optimizes the expected tightness between the induced lower bound and the ground-truth distance. The second, \textbf{hard-aware pivot selection}, focuses on preserving similarity ordering among hard (near-neighbor) trajectory pairs that are most critical for retrieval accuracy. Both algorithms are designed to operate efficiently on sampled training pairs and can be seamlessly integrated into the pivot-based framework.

A central challenge in pivot-based representations lies in selecting effective pivots, as poorly chosen pivots can result in loose lower bounds and degraded ranking performance. To address this issue, we propose two data-driven pivot selection algorithms tailored for trajectory similarity learning. The first, \textbf{tightness-aware pivot selection}, explicitly optimizes the expected tightness between the induced lower bound and the ground-truth distance. The second, \textbf{hardness-aware pivot selection}, emphasizes preserving similarity ordering among hard (i.e., near-neighbor) trajectory pairs that are most critical for retrieval accuracy. Both algorithms operate efficiently on sampled training pairs and can be integrated  seamlessly into the proposed pivot-based framework.

Experiments on multiple real-world trajectory datasets and across diverse trajectory distance measures offer evidence that the proposed lower-bound representations, when equipped with effective pivot selection strategies, are able to consistently outperform state-of-the-art neural trajectory embeddings in terms of ranking accuracy and stability. These findings suggest that carefully designed lower-bound representations provide an attractive alternative to neural embeddings for trajectory similarity learning, offering improved quality, interpretability, efficiency, and distance-agnostic applicability.

% Through extensive experiments on multiple real-world trajectory datasets and across multiple trajectory distance measures, we show that the proposed lower-bound representations, when equipped with effective pivot selection strategies, consistently outperform state-of-the-art neural trajectory embeddings in terms of ranking accuracy and stability. These results demonstrate that carefully designed lower-bound representations constitute a strong and often superior alternative to complex neural embeddings for trajectory similarity learning, while offering simplicity, interpretability, and distance-agnostic applicability.

The paper's contributions can be summarized as follows:
\begin{itemize}[leftmargin=*]
\item We adopt a lower-bound representation approach to trajectory similarity learning and propose \paramark{LB-TrajRep, a unified pivot-based framework that induces admissible, interpretable lower bounds for multiple classical trajectory distances, including DTW, Hausdorff distance, and DFD.}
% , without modifying the original distance definitions.
\item We propose two data-driven pivot selection algorithms, namely tightness-aware and hardness-aware selection, which explicitly optimize representation quality for similarity ranking and effectively target ranking-critical trajectory pairs, while avoiding the reliance on costly neural models.
\item We report on extensive experiments on multiple real-world trajectory datasets and covering diverse distance measures, finding that the proposed lower-bound representations consistently outperform state-of-the-art neural trajectory embeddings by significant margins in terms of ranking accuracy and stability under a unified retrieval protocol.
\end{itemize}

The remainder of this paper is organized as follows.
Section~\ref{sec:preliminary} introduces preliminaries and defines the problem addressed.
Section~\ref{sec:methodology} presents the proposed lower-bound representation framework and the corresponding pivot selection strategies for trajectory similarity retrieval.
Section~\ref{sec:experiments} reports on extensive experimental evaluations on real-world trajectory datasets, including comparisons with state-of-the-art neural embedding methods.
Section~\ref{sec:related_works} reviews related work, and Section~\ref{sec:conclusion} concludes and proposes future research directions.

\section{Preliminaries}
\label{sec:preliminary}

We introduce concepts that are important in the
context of trajectory similarity learning and then define the problem addressed.
{\color{black}{Table~\ref{tab:notation} summarizes key notation used throughout the paper.}}
\begin{myDef}[Trajectory]
A trajectory is an ordered sequence of spatial points $T = \langle x_1, x_2, \dots, x_{|T|} \rangle$, where each point $x_i \in \mathbb{R}^d$ represents a location in a $d$-dimensional space.
We focus on spatial trajectories with $d=2$, corresponding to latitude and longitude coordinates.
The length of a trajectory $T$ is denoted by $|T|$.
\end{myDef}
% \begin{myDef}[Trajectory Dataset]
% Let $\mathcal{D} = \{T_1, T_2, \dots, T_N\}$ denote a dataset of trajectories.
% We use $\mathcal{Q} \subseteq \mathcal{D}$ to denote a set of query trajectories.
% \end{myDef}
Given trajectories as the basic data objects, trajectory similarity is typically quantified using distance measures that capture different aspects of spatial and sequential similarity.
We consider three classical trajectory distance measures that capture different notions of similarity.
% In this work, we consider three classical trajectory distance measures that capture different notions of similarity.

\begin{table}[t]
{\color{black}{
\centering
\caption{\color{black}{Summary of notation.}}
\vspace{-0.3cm}
\label{tab:notation}
\small
\begin{tabular}{ll}
\toprule
\textbf{Notation} & \textbf{Description} \\
\midrule
$T,S,Q$ & Trajectories; $Q$ denotes a query trajectory \\
$D(\cdot,\cdot)$ & Ground-truth trajectory distance \\
$LB(\cdot,\cdot)$ & Lower bound of $D(\cdot,\cdot)$ \\
$C$ & Set of lower-bound components \\
$LB_C(T,S)$ & Lower bound induced by component set $C$ \\
$\phi_c(T)$ & Response of trajectory $T$ to component $c$ \\
$\Phi_C(T)$ & Lower-bound representation of trajectory $T$ \\
$p$ & Spatial pivot point \\
$P$ & Selected pivot set \\
% $\phi_p(T)$ & Pivot response of trajectory $T$ to pivot $p$ \\
$k$ & Pivot budget ($|P|=k$) \\
$\mathcal{T}_{\text{train}}$ & Training trajectory set \\
$\mathcal{S}$ & Sampled trajectory pairs \\
% $\Delta(p\mid P)$ & Marginal gain of adding pivot $p$ to pivot set $P$ \\
\bottomrule
\end{tabular}
}}
\end{table}

\begin{myDef}[Hausdorff Distance~\cite{huttenlocher2002comparing,atev2010clustering}]
Given two trajectories $T$ and $S$, their Hausdorff distance is defined as follows:
\begin{equation}
D_H(T,S) = \max \left(
\max_{x \in T} \min_{y \in S} \|x - y\|,
\;
\max_{y \in S} \min_{x \in T} \|x - y\|
\right),
\end{equation}
where $\|\cdot\|$ denotes Euclidean distance.
Hausdorff distance is a metric and measures the maximum spatial deviation between two trajectories.
\end{myDef}

% While Hausdorff distance focuses on spatial coverage, other distances additionally consider the ordering of points along trajectories.

\begin{myDef}[Discrete Fr\'echet Distance~\cite{alt1995computing}]
The discrete Fr\'echet distance (DFD) between two trajectories $T$ and $S$ measures their similarity while respecting the ordering of the points in each trajectory.
% It is defined as the minimum leash length required to traverse both trajectories from start to end while preserving point order.
It is defined as the minimum distance between pairs of trajectory points required to traverse the trajectories from start to end.
DFD can be computed using dynamic programming and satisfies the properties of a metric.
\end{myDef}

% In contrast to the above metric distances, Dynamic Time Warping allows flexible alignments between trajectories.

\begin{myDef}[Dynamic Time Warping Distance~\cite{yi1998efficient}]
Given two trajectories $T = \langle x_1,\dots,x_n\rangle$ and
$S = \langle y_1,\dots,y_m\rangle$, the Dynamic Time Warping (DTW) distance is defined as follows:
\begin{equation}
D_{\mathrm{DTW}}(T,S) = \min_{\pi} \sum_{(i,j)\in \pi} \|x_i - y_j\|,
\end{equation}
where $\pi$ denotes a warping path that aligns points from $T$ and $S$ while preserving their order.
DTW can be computed in $O(nm)$ time using dynamic programming.
Unlike Hausdorff distance and DFD, DTW does not satisfy the triangle inequality and is not a metric.
\end{myDef}

Given a trajectory distance measure, similarity search is typically conducted by ranking trajectories according to their distances to a query trajectory.

% \begin{myDef}[Trajectory Similarity Ranking]
% Given a query trajectory $Q$ and a set of trajectories $\mathcal{D}$, the similarity ranking induced by a distance measure $D(\cdot,\cdot)$ is defined as an ordering of trajectories in $\mathcal{D}$ sorted by increasing distance to $Q$, denoted as
% $\mathrm{rank}_D(Q, \mathcal{D}).$
% \end{myDef}

\begin{myDef}[Trajectory Similarity Ranking]
\label{def:trajectory_similarity_ranking}
Given a query trajectory $Q$ and a trajectory dataset $\mathcal{D}$, the ranking induced by a distance measure $D(\cdot,\cdot)$ is the ordering of trajectories in $\mathcal{D}$ by increasing distance to $Q$, denoted as $\text{rank}_D(Q, \mathcal{D})$.
\end{myDef}

% Trajectory similarity learning aims to approximate such rankings using efficient representations.

% \begin{myDef}[Trajectory Similarity Learning~\cite{chang2024revisiting}]
% Given a trajectory distance measure $D(\cdot,\cdot)$, trajectory similarity representation learning aims to construct a representation function $f(\cdot)$ such that the distance computed on representations reflects the similarity between trajectories under $D$.
% Formally, for any trajectories $T_i$ and $T_j$, the representation distance $d(f(T_i), f(T_j))$ is expected to correlate with the original distance $D(T_i, T_j)$.

% Moreover, to support similarity retrieval, the representation should preserve the relative ordering induced by $D$.
% That is, for any three trajectories $T_i$, $T_j$, and $T_k$, if
% \[
% D(T_i, T_j) < D(T_i, T_k),
% \]
% then it is desirable that
% \[
% d(f(T_i), f(T_j)) < d(f(T_i), f(T_k)).
% \]
% \end{myDef}

\begin{myDef}[Trajectory Similarity Learning~\cite{chang2024revisiting}] 
Given a trajectory similarity measure $D(\cdot,\cdot)$, trajectory similarity representation learning aims to construct a representation function $f(\cdot)$ such that the distance computed on representations reflects the similarity between the original trajectories under $D$.
Formally, for any trajectories $T_i$ and $T_j$, the distance $d(f(T_i), f(T_j))$ is expected to correlate with the original distance $D(T_i, T_j)$, e.g., $d(f(T_i), f(T_j))\approx D(T_i, T_j)$, where $d(\cdot, \cdot)$ is a vector similarity measure.
Moreover, to support similarity retrieval, the representation should preserve the similarity ranking induced by $D$: for all query trajectory $Q$ and a trajectory dataset $\mathcal{D}$, the ranking induced by distances computed on representations should approximate the ground-truth ranking $\mathrm{rank}_D(Q, \mathcal{D})$ (see Definition~\ref{def:trajectory_similarity_ranking}).
\end{myDef}

Although many existing methods emphasize approximating the original distance values, in practice, \paramark{trajectory similarity learning is predominantly evaluated using top-$k$ ranking quality metrics such as Hit Ratio (HR) and Recall}~\cite{yao2019computing,yao2022trajgat}. These metrics focus on whether the nearest neighbors under the original distance are retrieved correctly, rather than on the exact numerical accuracy of distance approximation. As a result, preserving similarity ranking is often emphasized over minimizing absolute distance errors.

% \begin{myDef}[Trajectory Similarity Learning]
% Given a trajectory distance measure $D(\cdot,\cdot)$, trajectory similarity learning seeks to learn a representation function $f(\cdot)$ such that the ranking induced by a simple distance computed on representations $f(T)$ approximates the ground-truth ranking $\mathrm{rank}_D(Q, \mathcal{D})$.
% \end{myDef}

% Preserving similarity ranking as defined above requires representations that reflect the structure of the original trajectory distance.
% One principled way to achieve this is to exploit distance bounds that are guaranteed with respect to the ground-truth distance.
% In particular, lower bounds provide a natural mechanism to approximate similarity ordering while maintaining theoretical consistency with the original distance measure.

% Existing trajectory similarity learning methods typically address the above objective by learning representations whose pairwise distances approximate the original trajectory distance through data-driven optimization.
% Such approaches rely on regression or ranking-based losses and do not provide explicit guarantees with respect to the ground-truth distance.
% In contrast, we adopt a different perspective and consider representations derived from distance bounds.
% Lower bounds offer a principled mechanism to relate representation distances to the original trajectory distance, making them a natural foundation for constructing similarity-preserving representations.

Existing trajectory similarity learning methods address this objective by learning neural embedding representations whose pairwise distances are trained to approximate the original trajectory distance or its induced ranking~\cite{chang2024revisiting,yang2024simformer,chang2024trajectory}. These methods rely on complex neural models and optimization objectives, and they typically do not provide explicit guarantees with respect to the ground-truth distance. In contrast, we adopt a different perspective and consider representations derived from distance bounds. Lower bounds offer a principled mechanism to relate representation distances to the original trajectory distance, making them a natural foundation for constructing similarity-preserving representations with theoretical consistency guarantees.

% Lower bounds play an important role in trajectory similarity search and ranking by providing guaranteed distance approximations.

\begin{myDef}[Lower Bound of a Trajectory Distance]
Given a trajectory distance measure $D(\cdot,\cdot)$, a function $LB(\cdot,\cdot)$ is called a lower bound of $D$ if it satisfies $ \forall\ T,S, LB(T,S) \le D(T,S)$.
\end{myDef}

The effectiveness of a lower bound for ranking preservation depends on how closely it approximates the true distance.

\begin{myDef}[Lower-Bound Tightness]
\label{def:tightness}
The tightness of a lower bound $LB(\cdot,\cdot)$ with respect to a distance $D(\cdot,\cdot)$ denotes how well $LB(T,S)$ approximates $D(T,S)$.
Tighter lower bounds generally preserve similarity rankings better.
\end{myDef}

% We are now ready to formally define the problem addressed in this work.
While lower-bound tightness characterizes how closely a lower bound approximates the original distance, it is not an end goal in itself.
Rather tightness affects the quality of similarity rankings induced by lower-bound–based representations.
Motivated by this observation, we formulate the problem of trajectory similarity learning as constructing representations that leverage lower bounds to preserve the ground-truth similarity ranking.

% \begin{myDef}[Problem Definition]
% Given a trajectory dataset $\mathcal{D}$, a set of query trajectories $\mathcal{Q}$, and a trajectory distance measure $D(\cdot,\cdot)$, the problem is to construct a representation for each trajectory such that the similarity ranking induced by a lower bound computed on the representations closely approximates the ground-truth ranking $\mathrm{rank}_D(Q, \mathcal{D})$ for any query $Q \in \mathcal{Q}$.
% \end{myDef}

% \begin{myDef}[Problem Definition]
% Given a trajectory dataset $\mathcal{D}$, a set of query trajectories $\mathcal{Q}$, and a trajectory distance measure $D(\cdot,\cdot)$, 
% the goal of this work is to construct a representation for each trajectory such that similarity between trajectories can be computed efficiently using a simple distance function on the representations.
% At the same time, the induced similarity ranking should closely approximate the ground-truth ranking $\mathrm{rank}_D(Q, \mathcal{D})$ for any query $Q \in \mathcal{Q}$.
% In particular, we focus on representations derived from lower bounds of $D$, which provide theoretical consistency with the original trajectory distance while enabling fast similarity computation in a retrieval-oriented setting.
% \end{myDef}

\begin{myDef}[Problem Definition]
Given a trajectory dataset $\mathcal{D}$, a query set $\mathcal{Q}$, and a trajectory distance measure $D(\cdot,\cdot)$, the goal is to construct a representation function $f(\cdot)$ that enables efficient similarity computation using simple vector distances, while approximately preserving the ground-truth similarity ranking $\text{rank}_D(Q, \mathcal{D})$ for all $Q \in \mathcal{Q}$. In particular, we focus on representations derived from lower bounds of $D$, which makes it possible to provide theoretical guarantees on the consistency with the original distance in a retrieval-oriented setting.
\end{myDef}

\section{Methodology}
\label{sec:methodology}

{\color{black}{We proceed to present LB-TrajRep, a unified framework for trajectory similarity learning based on lower-bound representations, as shown in Figure~\ref{fig:framework}.}}
Rather than learning embeddings using neural network based methods, our approach constructs representations from lower-bound components, thereby enabling theoretical guarantees with respect to the original trajectory distance.
Specifically, a trajectory is represented using a small set of lower-bound components, enabling efficient similarity computation while effectively preserving ground-truth similarity rankings.
We first introduce the pivot-based lower-bound representation framework, then describe different component instantiations, and finally present data-driven pivot selection strategies for improving ranking quality.

\subsection{Lower-Bound Representation Framework}
\label{sec:Lower-Bound-Representation-Framework}

% We begin by introducing a unified lower-bound representation framework for trajectory similarity learning.
The key idea underlying the proposed representation framework is to represent each trajectory using a set of lower-bound components derived from simple geometric primitives, so that the similarity between trajectories can be estimated efficiently while maintaining theoretical consistency with the original trajectory distance.

Let $\mathcal{C}$ denote a set of lower-bound components.
Each component $c \in \mathcal{C}$ induces a valid lower bound $LB_c(\cdot,\cdot)$ of a given trajectory distance $D(\cdot,\cdot)$.
Given two trajectories $T$ and $S$, the overall lower bound induced by $\mathcal{C}$ is defined as follows:
\begin{equation}
LB_{\mathcal{C}}(T,S) = \max_{c \in \mathcal{C}} LB_c(T,S)
\end{equation}
By construction, $LB_{\mathcal{C}}(T,S)$ is also a valid lower bound of $D(T,S)$.

This approach associates a representation with each trajectory that consists of its responses to the components in $\mathcal{C}$.
Specifically, for each component $c$, we define a scalar function $\phi_c(T)$ such that:
\begin{equation}
LB_c(T,S) = |\phi_c(T) - \phi_c(S)|
\end{equation}
This formulation applies to lower-bound components induced by scalar projections, which cover all instantiations considered in this work.
As a result, a trajectory $T$ can be represented as a vector:
\begin{equation}
\Phi_{\mathcal{C}}(T) = \bigl[\phi_{c_1}(T), \phi_{c_2}(T), \dots, \phi_{c_{|\mathcal{C}|}}(T)\bigr],
\end{equation}
where $\mathcal{C}=\{c_1,\dots,c_{|\mathcal{C}|}\}$.
Similarity between trajectories can then be computed directly on these representations using simple distance functions, such as the $\ell_\infty$ distance, which corresponds to the maximum over component-wise differences.

Different instantiations of the framework correspond to different choices of the component set $\mathcal{C}$.
Next, we describe concrete lower-bound components and their instantiations, and we study how to select effective components to improve similarity ranking quality.

\subsection{Lower-Bound Components and Instantiations}

Under the above lower-bound representation framework, the quality and behavior of a representation are determined by the choice of lower-bound components.
We introduce two fundamental types of lower-bound components and show how combinations of these components lead to different instantiations of the framework.

% \textbf{Endpoint-Based Components.}
% A simple yet effective source of lower bounds arises from the start and end points of trajectories.
% Let $T = \langle x_1, \dots, x_{|T|} \rangle$ and $S = \langle y_1, \dots, y_{|S|} \rangle$ be two trajectories.
% For many trajectory distances that respect point-wise alignment or traversal order, the distances between corresponding endpoints provide valid lower bounds.
% Specifically, we define the endpoint-based lower bound as
% \begin{equation}
% \label{eq:LB_SE}
% LB_{\mathrm{SE}}(T,S) = \max \left\{ \|x_1 - y_1\|,\; \|x_{|T|} - y_{|S|}\| \right\}.
% \end{equation}
% This component is inexpensive to compute and captures coarse spatial alignment between trajectories.
% It serves as a strong baseline and is commonly used in classical trajectory similarity search~\cite{Keogh2002ExactIO}.

\begin{figure}[t]
    \centering
    \includegraphics[width=\columnwidth]{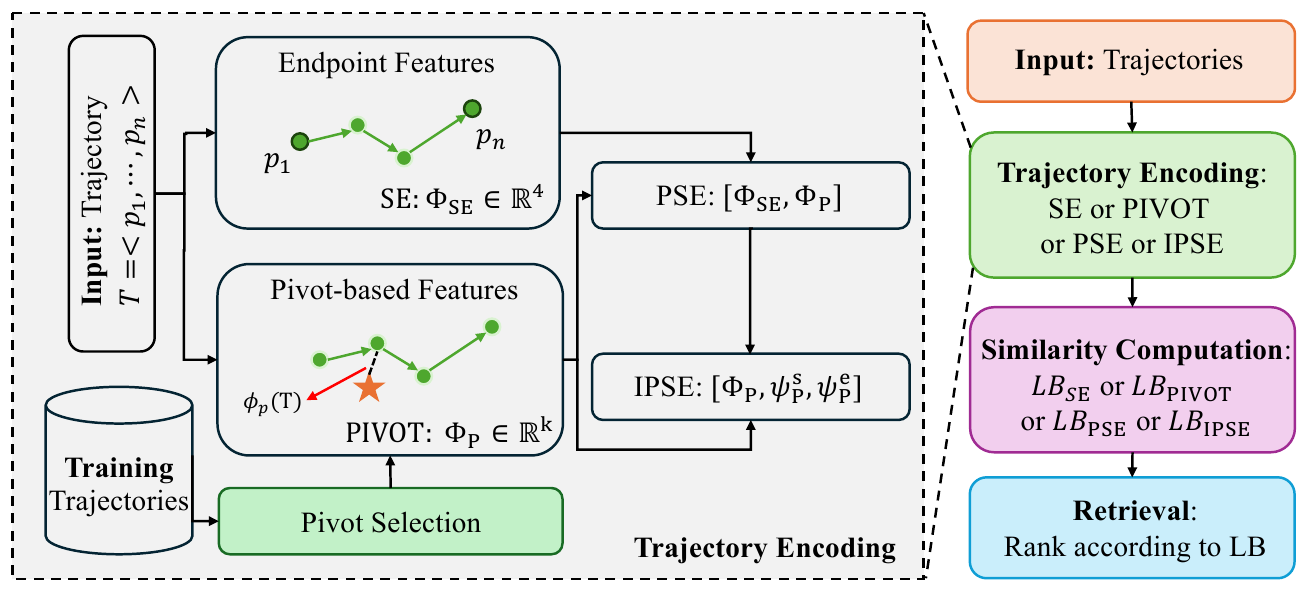}
    \vspace{-0.7cm}
    % \caption{The proposed framework.}
    \caption{\color{black}{LB-TrajRep framework overview.}}
    \label{fig:framework}
    \vspace{-0.5cm}
\end{figure}

\textbf{Endpoint-Based Components.}
A simple yet effective source of lower bounds arises from the start and end points of trajectories.
% Let $T = \langle x_1, \dots, x_{|T|} \rangle$ and $S = \langle y_1, \dots, y_{|S|} \rangle$ be two trajectories.
For many trajectory distances that enforce point-wise alignment or respect traversal order, the distances between corresponding endpoints provide valid lower bounds.
Specifically, given trajectories $T = \langle x_1, \dots, x_{|T|} \rangle$ and $S = \langle y_1, \dots, y_{|S|} \rangle$, we define the endpoint-based lower bound:
\begin{equation}
\label{eq:LB_SE}
LB_{\mathrm{SE}}(T,S) = \max \left( \|x_1 - y_1\|,\; \|x_{|T|} - y_{|S|}\| \right)
\end{equation}
This component is inexpensive to compute and captures coarse spatial alignment between trajectories.
It serves as a strong baseline and is commonly used in classical trajectory similarity search~\cite{Keogh2002ExactIO}.
However, this endpoint-based component provides a valid lower bound for alignment-based distances such as DTW and DFD, but not for Hausdorff distance.
% , as the maximum deviation under Hausdorff distance may occur at interior points rather than at trajectory endpoints.

\textbf{Pivot-Based Components.}
% Beyond endpoints, we consider lower-bound components derived from a set of pivot points in the spatial domain.
% Importantly, pivots are defined as points rather than trajectories, which allows the resulting components to remain applicable to both metric and non-metric trajectory distances.
Next, we consider lower-bound components derived from a set of pivot points in the spatial domain. 
% While pivot-based lower bounds have been studied in prior work, existing approaches primarily focus on metric spaces and are designed for distance pruning or index-based acceleration.
Pivot-based lower bounds have been considered for metric indexing and distance pruning~\cite{Yianilos1993DataSA,Chen2020IndexingMS,Chvez2001SearchingIM}, where approaches are typically designed for metric distances and rely on trajectory-level pivots~\cite{Zheng2021REPOSEDT}.
% In contrast, we adopt pivot-based components as a representation learning mechanism, with pivots defined as points rather than trajectories, which enables applicability to both metric and non-metric trajectory distances such as DTW. 
In contrast, we consider pivot-based lower bounds from a trajectory similarity learning perspective and redefine pivots as points rather than trajectories. This seemingly simple design choice fundamentally differs from prior pivot-based approaches and allows the resulting components to remain valid for both metric and non-metric trajectory distances, including DTW (as covered in Section~\ref{sec:pivot_based_lower_bounds_for_metric_and_nonmetric_distances}). 
In addition, defining pivots as points rather than trajectories substantially simplifies representation construction. At query time, encoding a trajectory only requires computing point-to-trajectory distances, avoiding the need for expensive trajectory-to-trajectory distance computations that are commonly required by trajectory-level pivots.

Given a pivot point $p \in \mathbb{R}^2$ and a trajectory $T$, we define a pivot response function:
\begin{equation}
\phi_p(T) = \min_{x \in T} \|x - p\|
\end{equation}
Intuitively, $\phi_p(T)$ measures the minimum distance from trajectory $T$ to the pivot point $p$.
For two trajectories $T$ and $S$, the corresponding pivot-based lower bound is defined as follows:
\begin{equation}
LB_p(T,S) = |\phi_p(T) - \phi_p(S)|
\end{equation}
Each pivot thus induces a valid lower-bound component that captures the relative proximity between trajectories with respect to a reference location in space. The effectiveness of pivot-based components depends on the choice of pivots, which motivates the data-driven pivot selection strategies presented in Section~\ref{sec:pivot_selection_for_lower_bound_representations}.

% Different instantiations of the lower-bound representation framework correspond to different choices and combinations of the above components. We consider four representative instantiations.
{\color{black}{\textbf{Framework Instantiations.} The lower-bound representation framework introduced in Section~\ref{sec:Lower-Bound-Representation-Framework} is generic and is open to different choices of the component set $\mathcal{C}$.
Using the endpoint-based component $LB_{SE}$ and the pivot-based components $\{LB_p | p \in P\}$, we instantiate the framework in four different ways, corresponding to different choices of $\mathcal{C}$ and different ways of constructing trajectory representations. Table~\ref{tab:inst} summarizes the four instantiations.}}

% \textbf{SE:} This instantiation uses only the endpoint-based component $LB_{\mathrm{SE}}$.
% It serves as a lightweight baseline that relies solely on trajectory endpoints.

% \underline{\textsf{SE}}: This instantiation uses only the endpoint-based component
% $LB_{SE}$. Each trajectory is represented by a $4$-dimensional vector consisting of the coordinates of its start and end points. The similarity between trajectories is computed by applying the representation distance corresponding to Eq.~\ref{eq:LB_SE}. This instantiation serves as an extremely compact and efficient baseline.

{\color{black}{\underline{\textsf{SE}}:
This instantiation uses only the endpoint-based lower-bound component $LB_{\mathrm{SE}}$ defined in Eq.~\ref{eq:LB_SE}.
This corresponds to the component set containing only the start--end component.
Each trajectory is represented by a 4-dimensional vector consisting of the coordinates of its start and end points.
}}
% The similarity between trajectories is computed by the endpoint-based lower bound $LB_{\mathrm{SE}}(T,S)$.
This instantiation serves as an extremely compact and efficient baseline.

% \begin{figure}[t]
%     \centering
%     \includegraphics[width=\columnwidth]{figs/framework.pdf}
%     % \vspace{-0.7cm}
%     \caption{The proposed framework.}
%     \label{fig:framework}
%     % \vspace{-0.5cm}
% \end{figure}

\begin{table}[t]
\centering
{\color{black}{
\caption{\color{black}{Framework instantiations.}}
\vspace{-0.3cm}
\label{tab:inst}
\small
\begin{tabular}{lll}
\toprule
Method & Component Set $\mathcal{C}$ & Representation \\
\midrule
SE
&
$\{LB_{SE}\}$
&
Endpoints
\\

PIVOT
&
$\{LB_p\}_{p\in P}$
&
Pivot responses
\\

PSE
&
$\{LB_{SE}\}\cup\{LB_p\}_{p\in P}$
&
Endpoints \& pivot responses
\\

IPSE
&
$\{LB_p\}_{p\in P}$
&
Pivot responses \& endpoint-to-pivot
\\
\bottomrule
\end{tabular}
}}
\vspace{-0.3cm}
\end{table}

% , yielding a valid lower bound for alignment-based distances

% \textbf{PIVOT:} This instantiation uses a set of pivot-based components $\{LB_p \mid p \in P\}$ induced by a pivot set $P$, without incorporating endpoint information.

% \underline{\textsf{PIVOT}}: This instantiation uses a set of pivot-based components $\{LB_p \mid p \in P\}$ induced by a pivot set $P$, without incorporating endpoint information. Each trajectory $T$ is represented by the pivot-response vector $\phi_P(T) = (\phi_{p_1}(T), \ldots, \phi_{p_k}(T))$, where $P=\{p_1,\ldots,p_k\}$ and $\phi_p(T)=\min_{x\in T}\|x-p\|$. The similarity between trajectories is computed directly on these
% representations using a simple distance function, i.e., $\ell_\infty$.

{\color{black}{
\underline{\textsf{PIVOT}}: This instantiation uses only the pivot-based components induced by a pivot set $P$.
The induced lower bound is defined as follows:
\begin{equation}
LB_{\mathrm{PIVOT}}(T,S)
=
\max_{p\in P} LB_p(T,S)
=
\max_{p\in P}|\phi_p(T)-\phi_p(S)|.
\label{eq:pivot}
\end{equation}
Each trajectory $T$ is represented by the pivot-response vector
$\phi_P(T)=(\phi_{p_1}(T),\ldots,\phi_{p_k}(T))$, where
$P=\{p_1,\ldots,p_k\}$ and
$\phi_p(T)=\min_{x\in T}$ $\|x-p\|$.
Trajectory similarity is computed directly on these representations using the $\ell_\infty$ distance.
}}

\underline{\textsf{PSE}}: This instantiation combines pivot-based components with the endpoint-based component.
The induced lower bound is defined as follows:
{\color{black}{
\begin{equation}
LB_{\mathrm{PSE}}(T,S) = \max \left( LB_{\mathrm{SE}}(T,S),\; LB_{\mathrm{PIVOT}}(T,S) \right)
\end{equation}
}}
% This explicit combination preserves the admissibility of the lower bound while leveraging complementary information from endpoints and pivots.
This combination integrates complementary information from endpoints and pivots and can be applied whenever the underlying components constitute valid lower bounds on the target distance.

{\color{black}{\underline{\textsf{IPSE}}:
Unlike PSE, which combines lower-bound values via a maximum operator, IPSE incorporates endpoint information at the representation level.
Specifically, endpoint information is encoded by the distances from the start and end points to the selected pivots, and these features are concatenated with the pivot-response vector.
Let $P=\{p_1,\ldots,p_k\}$ be the selected pivot set and let $T=\langle x_1,\ldots,x_{|T|}\rangle$.
We reuse the pivot-response vector $\phi_P(T)$ defined above and define the endpoint-to-pivot distance vectors as follows:
}}
\begin{equation}
\begin{aligned}
\psi^{\mathrm{s}}_P(T)&=\bigl[\|x_1-p_1\|,\ldots,\|x_1-p_k\|\bigr] \\
\psi^{\mathrm{e}}_P(T)&=\bigl[\|x_{|T|}-p_1\|,\ldots,\|x_{|T|}-p_k\|\bigr]
\end{aligned}
\end{equation}
The IPSE representation is the concatenation of the above three vectors:
\begin{equation}
\Phi_{\mathrm{IPSE}}(T)=\bigl[\phi_P(T)\ \Vert\ \psi^{\mathrm{s}}_P(T)\ \Vert\ \psi^{\mathrm{e}}_P(T)\bigr]
\end{equation}
% The similarity between trajectories is then computed uniformly on $\Phi_{\mathrm{IPSE}}(\cdot)$ using a simple distance function.
% Unlike PSE, which combines endpoint- and pivot-based lower bounds via a maximum operator, IPSE integrates endpoint information directly into the representation, yielding a unified vector form.
{\color{black}{
Trajectory similarity is then computed directly on $\Phi_{\mathrm{IPSE}}(\cdot)$ using the $\ell_\infty$ distance, i.e., $LB_{\mathrm{IPSE}}(T,S) = \|\Phi_{\mathrm{IPSE}}(T)-\Phi_{\mathrm{IPSE}}(S)\|_\infty$. For DTW and DFD, this distance also constitutes a valid lower bound, and it is also compatible with the representation-based paradigm commonly adopted in trajectory similarity learning.
}}

\textbf{Discussion.} All instantiations described above are derived from the same lower-bound representation framework and differ only in the choice and combination of lower-bound components.
As a result, they share the same theoretical guarantees with respect to the original trajectory distance, subject to the validity of the underlying lower-bound components.
% As a result, they share the same theoretical guarantees with respect to the original trajectory distance, while exhibiting different trade-offs between representation compactness and ranking effectiveness.
Moreover, an important property of defining pivots as points is that the resulting trajectory representation depends only on the selected pivot set, not on the specific trajectory distance.
Once pivot points are fixed, the same representation vector can be reused to induce valid lower bounds for different trajectory distances, including Hausdorff distance, DFD, and DTW.
Switching the target distance only affects how similarity is interpreted on the representation, without requiring re-encoding of trajectories.

\subsection{Pivot Selection for Lower-Bound Representations}
\label{sec:pivot_selection_for_lower_bound_representations}

The effectiveness of the proposed pivot-based representations depends on the choice of pivots, and different pivot sets can induce substantially different lower bounds, which in turn affect the quality of the resulting similarity rankings. 
Existing pivot selection strategies are primarily developed for metric indexing and lower-bound--based similarity search, where pivots are selected from the same domain as the indexed objects and are often guided by global geometric criteria such as spatial coverage or dispersion~\cite{Yianilos1993DataSA,Chvez2001SearchingIM}. 

In our setting, pivots are instantiated as spatial points while data objects are trajectories, leading to a heterogeneous pivot--data relationship. In this setting, global coverage alone does not fully capture how a pivot contributes to representation quality, as the effectiveness of a pivot also depends on how tightly it bounds distances between similar trajectories when used as a representation component. This observation motivates the development of pivot selection strategies that explicitly account for lower-bound behavior in the representation space.
Before considering pivot selection, we clarify why the proposed point-based pivots induce valid lower bounds for both metric and non-metric trajectory distances.

\subsubsection{Pivot-Based Lower Bounds for Metric and Non-Metric Distances}
\label{sec:pivot_based_lower_bounds_for_metric_and_nonmetric_distances}

We first clarify the theoretical basis of pivot-based lower bounds and explain why defining pivots as points, rather than trajectories, makes it possible to support non-metric trajectory distances such as DTW.

\textbf{Pivot-based lower bounds in metric spaces.}
Pivot-based lower bounds are a classical technique for metric distances~\cite{Chen2020IndexingMS,chen2015efficient}. Let $D(\cdot,\cdot)$ be a metric trajectory distance and let $P$ be a set of pivots drawn from the same domain as the data objects. For any pivot $P_i$ and any two trajectories $T$ and $S$, the reverse triangle inequality implies
\begin{equation}
|D(T,P_i) - D(S,P_i)| \le D(T,S),
\end{equation}
which yields a valid lower bound. Most existing pivot-based methods implicitly rely on this property and therefore select pivots that are homogeneous with the data objects, i.e., trajectories themselves. This formulation is applicable to metric distances such as Hausdorff distance and DFD.
% Discrete Fr\'echet Distance.

\textbf{Limitations for non-metric distances.}
For non-metric distances such as DTW, the triangle inequality does not hold. As a result, trajectory-level pivoting does not, in general, induce valid lower bounds. This limitation prevents direct application of existing pivot-based techniques to DTW and other non-metric trajectory distances.

\textbf{Point pivots and non-metric lower bounds.}
To overcome this limitation, we define pivots as spatial points rather than trajectories. Given a pivot point $p\in\mathbb{R}^2$, the pivot response $\phi_p(T)=\min_{x\in T}\|x-p\|$ induces a scalar projection of a trajectory. Importantly, this projection does not rely on any metric property of the trajectory distance itself.

% \begin{lemma}
% For any pivot point $p$ and any two points $x,y\in\mathbb{R}^d$, the mapping $g_p(z)=\|z-p\|$ is 1-Lipschitz, i.e.,
% \begin{equation}
% |g_p(x)-g_p(y)| \le \|x-y\|.
% \end{equation}
% \end{lemma}

\begin{lemma}[Lipschitz Projection Induced by Point Pivots]
\label{lem:lipschitz}
For any pivot point $p \in \mathbb{R}^2$ and any two points $x, y \in \mathbb{R}^2$, the mapping $g_p(z) = \|z - p\|$ is 1-Lipschitz, i.e., $|g_p(x) - g_p(y)| \le \|x - y\|$.
% \begin{equation}
% |g_p(x) - g_p(y)| \le \|x - y\|.
% \end{equation}
\end{lemma}

\begin{lemma}[Trajectory-level lower bound induced by point pivots]
\label{lem:pointpivot_lb}
Let $p\in\mathbb{R}^2$ be a pivot point and define the pivot response of a trajectory $T$ as $\phi_p(T)=\min_{x\in T}\|x-p\|=\min_{x\in T} g_p(x)$.
Consider any trajectory distance $D(\cdot,\cdot)$ that \emph{dominates pointwise distances} in the sense that
for any two trajectories $T,S$:
\begin{equation}
\label{eq:dominates_pointwise_distances}
\forall x\in T:\ \min_{y\in S}\|x-y\|\le D(T,S),
\ \ \
\forall y\in S:\ \min_{x\in T}\|x-y\|\le D(T,S)
\end{equation}
Then for any trajectories $T$ and $S$, the following lower bound holds:
\begin{equation}
\label{eq:holds_lower_bound}
|\phi_p(T)-\phi_p(S)| \le D(T,S)
\end{equation}
\end{lemma}

\begin{proof}
{\color{black}{Please see the artifact repository for the detailed proof.}}
% Fix $p\in\mathbb{R}^2$ and let $g_p(z)=\|z-p\|$ as in Lemma~\ref{lem:lipschitz}.
% Let $x^\star = \arg\min_{x\in T} g_p(x)$, and $y^\star = \arg\min_{y\in S} g_p(y)$, so that $\phi_p(T)=g_p(x^\star)$ and $\phi_p(S)=g_p(y^\star)$.
% We first upper bound $\phi_p(T)-\phi_p(S)$.
% For any $y\in S$, by the triangle inequality, we have:
% \begin{equation}
% g_p(x^\star)=\|x^\star-p\| \le \|x^\star-y\|+\|y-p\|=\|x^\star-y\|+g_p(y).
% \end{equation}
% Rearranging terms yields $g_p(x^\star)-g_p(y)\le \|x^\star-y\|$
% % \begin{equation}
% % g_p(x^\star)-g_p(y)\le \|x^\star-y\|.
% % \end{equation}
% Taking the minimum over $y\in S$ on both sides gives
% \begin{equation}
% g_p(x^\star)-\min_{y\in S}g_p(y)\le \min_{y\in S}\|x^\star-y\|
% \end{equation}
% Using $\min_{y\in S}g_p(y)=\phi_p(S)$, we obtain:
% \begin{equation}
% \phi_p(T)-\phi_p(S)\le \min_{y\in S}\|x^\star-y\|
% \end{equation}
% By the domination assumption (applied to $x^\star\in T$), $\min_{y\in S}\|x^\star-y\|\le D(T,S)$, hence:
% \begin{equation}
% \phi_p(T)-\phi_p(S)\le D(T,S)
% \end{equation}
% Similarly, we can upper bound $\phi_p(S)-\phi_p(T)$ symmetrically. Combining the two inequalities gives $|\phi_p(T)-\phi_p(S)|\le D(T,S)$.
\end{proof}

The domination condition in Eq.~\ref{eq:dominates_pointwise_distances} is mild and is satisfied by the classical trajectory distances considered in this study.
For Hausdorff distance, it follows directly from its max--min definition, which explicitly takes the maximum over point-to-trajectory nearest-neighbor distances.
For alignment-based distances such as DTW, although they are non-metric, every point participates in the alignment and contributes a non-negative matching cost, implying that the minimum distance from any point to the other trajectory cannot exceed the total alignment cost.
DFD satisfies the condition for similar reasons.
Therefore, Eq.~\ref{eq:holds_lower_bound} holds for Hausdorff distance, DFD, and DTW.

% \noindent\emph{Proof sketch (DTW as an example).}
% Consider $T=\langle x_1,\ldots,x_n\rangle$ and $S=\langle y_1,\ldots,y_m\rangle$. Let $\pi$ be any DTW warping path, and denote its DTW cost under the original pointwise distance by
% \begin{equation}
% \mathrm{cost}_{\pi}(T,S)=\sum_{(i,j)\in\pi}\|x_i-y_j\|.
% \end{equation}
% Applying a 1-Lipschitz mapping $g(\cdot)$ yields $|g(x_i)-g(y_j)|\le \|x_i-y_j\|$ for every aligned pair $(i,j)\in\pi$, and thus
% \begin{equation}
% \mathrm{cost}_{\pi}\!\left(g(T),g(S)\right)
% =\sum_{(i,j)\in\pi}|g(x_i)-g(y_j)|
% \le \sum_{(i,j)\in\pi}\|x_i-y_j\|
% =\mathrm{cost}_{\pi}(T,S).
% \end{equation}
% Since DTW takes the minimum cost over all warping paths, we have
% \begin{equation}
% D_{\mathrm{DTW}}\!\left(g(T),g(S)\right)
% =\min_{\pi}\mathrm{cost}_{\pi}\!\left(g(T),g(S)\right)
% \le \min_{\pi}\mathrm{cost}_{\pi}(T,S)
% = D_{\mathrm{DTW}}(T,S).
% \end{equation}
% The same monotonicity argument applies to Hausdorff distance and DFD, as they are also defined by max/min aggregations over pointwise distances.

{
\begin{algorithm}[t]
\small
\caption{Tightness-Aware Point-Pivot Selection}
\label{alg:tightness-pivot}

\KwIn{training trajectories $\mathcal{T}_{\mathit{train}}$;
ground-truth distance access $D(\cdot,\cdot)$ on $\mathcal{T}_{\mathit{train}}$;
pivot budget $k$;
number of sampled trajectory pairs $M$;
candidate pool size $B$;
small constant $\epsilon>0$}
\KwOut{selected pivot point set $P$}

% by sampling points 
Build a point candidate pool $U$ from trajectories in $\mathcal{T}_{\mathit{train}}$\; \label{alg:line_1}
Sample $M$ trajectory pairs $\mathcal{S}=\{(T_a,T_b)\}$ from $\mathcal{T}_{\mathit{train}}$\; \label{alg:line_2}
\ForEach{$(T_a,T_b)\in\mathcal{S}$}{ \label{alg:line_3}
    $d(T_a,T_b)\leftarrow D(T_a,T_b)+\epsilon$; $c(T_a,T_b)\leftarrow 0$\; \label{alg:line_4}
}
$P\leftarrow\emptyset$; $R\leftarrow U$\; \label{alg:line_5}

\For{$i\leftarrow 1$ \KwTo $k$}{ \label{alg:line_6}
    Sample a subset $C\subseteq R$ where $|C|=\min\{B,|R|\}$\; \label{alg:line_7}
    % Select $p^\star\in C$ that maximizes the average marginal gain by updating $c(T_a,T_b)$ to
    $p^\star\leftarrow\max(c(T_a,T_b),|\phi_p(T_a)-\phi_p(T_b)|)/d(T_a,T_b)$ over all $(T_a, T_b)\in \mathcal{S}$\; \label{alg:line_8}
    $P\leftarrow P\cup\{p^\star\}$; $R\leftarrow R\setminus\{p^\star\}$\; \label{alg:line_9}
    \ForEach{$(T_a,T_b)\in\mathcal{S}$}{ \label{alg:line_10}
        $c(T_a,T_b)\leftarrow
        \max(c(T_a,T_b),|\phi_{p^\star}(T_a)-\phi_{p^\star}(T_b)|)$\; \label{alg:line_11}
    }
}
\Return{$P$}\; \label{alg:line_12}
\end{algorithm}
}

\subsubsection{Tightness-Aware Pivot Selection}
\label{sec:tightness-aware}

% Given a candidate pivot set $\mathcal{P}$ and a pivot budget $k$, our goal is to select a subset $P \subseteq \mathcal{P}$ that yields tight lower bounds for trajectory distance estimation. Rather than optimizing worst-case tightness, we focus on the \emph{average relative tightness} over a sampled set of trajectory pairs, which is more aligned with similarity ranking tasks.

We now study how to select pivots that optimize the quality of the induced lower-bound representations. Given a candidate pivot set $\mathcal{P}$ and a target pivot budget $k$, our goal is to select a subset $P\subseteq \mathcal{P}$, $|P|=k$, such that the induced lower bounds are as tight as possible for trajectory similarity learning.

\textbf{Tightness-aware objective.}
% Let $\mathcal{S}$ be a set of sampled trajectory pairs $(T,S)$ with ground-truth distance $D(T,S)$. For a pivot set $P$, recall that the pivot-based lower bound is
% \begin{equation}
% LB_P(T,S) = \max_{p \in P} \left| \phi_p(T) - \phi_p(S) \right|,
% \end{equation}
% where $\phi_p(T)=\min_{x\in T}\|x-p\|$ denotes the response of trajectory $T$ to pivot point $p$.
% We define the tightness-aware objective as
% \begin{equation}
% \label{eq:tightness-objective}
% F(P) = 
% \mathbb{E}_{(T,S)\sim \mathcal{S}}
% \left[
% \frac{LB_P(T,S)}{D(T,S)}
% \right].
% \end{equation}
% Pairs with $D(T,S)=0$ are excluded from $\mathcal{S}$, or equivalently handled by adding a small $\epsilon$ to the denominator in practice.
For a pivot set $P$, recall that the induced lower bound between trajectories $T$ and $S$ is $ LB_P(T,S)=\max_{p\in P} f_p(T,S),$ where $f_p(T,S)=|\phi_p(T)-\phi_p(S)|$.
Given a set of sampled trajectory pairs $\mathcal{S}$, we define the following
tightness-aware objective:
% \begin{equation}
% F(P)=
% \mathbb{E}_{(T,S)\sim \mathcal{S}}
% \left[
% \frac{LB_P(T,S)}{D(T,S)}
% \right].
% \end{equation}
\begin{equation}
\label{eq:tightness-objective}
\max_{P \subseteq \mathcal{P},\, |P| = k}
F(P) =
\mathbb{E}_{(T,S)\sim\mathcal{S}}
\left[\frac{LB_P(T,S)}{D(T,S)}\right]
\end{equation}
By normalizing the lower bound with the ground-truth distance, this objective measures how well the induced representation preserves relative similarity ordering across trajectory pairs, which directly aligns with ranking-based evaluation metrics.
Maximizing $F(P)$ encourages the selected pivots to produce tight lower bounds relative to the original distance, which benefits similarity ranking. 
\textbf{Optimization via greedy selection.}
{\color{black}{Eq.~\ref{eq:tightness-objective} states a cardinality-constrained submodular maximization problem, which is NP-hard in general~\cite{nemhauser1978analysis}.
% We next describe how this problem can be approximated efficiently using a greedy approach. 
% {\color{black}{The key idea is to iteratively select pivot points that maximize the improvement in lower-bound tightness. Specifically, we seek a pivot set whose induced lower bounds closely approximate the corresponding ground-truth trajectory distances over a collection of sampled trajectory pairs.}}
We thus proceed to approximate this problem efficiently using a greedy approach. The key idea is to iteratively select pivot points that maximize the improvement in lower-bound tightness. Starting from an empty pivot set, the algorithm repeatedly evaluates candidate pivots according to their marginal gains and adds the most informative pivot to the current pivot set.}}

{\color{black}{Algorithm~\ref{alg:tightness-pivot} implements this idea using a sampled set of trajectory pairs. It first constructs a candidate pool of pivot points and samples trajectory pairs to estimate the tightness objective (lines~\ref{alg:line_1}--\ref{alg:line_2}). For each sampled pair, the ground-truth distance is computed once and is then used to initialize the current lower-bound value (lines~\ref{alg:line_3}--\ref{alg:line_5}). During each iteration, the algorithm samples a subset of candidate pivots and evaluates their marginal gains on the sampled trajectory pairs (lines~\ref{alg:line_7}--\ref{alg:line_8}). The pivot with the largest gain is selected, and the corresponding lower-bound values are updated (lines~\ref{alg:line_9}--\ref{alg:line_11}). This process continues until the pivot budget is reached, after which the selected pivot set is returned (line~\ref{alg:line_12}).}}

% Specifically, starting from an empty set $P_0=\emptyset$, at each iteration $i$,
% we select
% \begin{equation}
% p_i = \arg\max_{p\in \mathcal{P}\setminus P_{i-1}}
% \left( F(P_{i-1}\cup\{p\}) - F(P_{i-1}) \right),
% \end{equation}
% and update $P_i=P_{i-1}\cup\{p_i\}$ until $|P_k|=k$.

The greedy procedure described in Algorithm~\ref{alg:tightness-pivot} admits a formal approximation guarantee, as stated below.

\begin{theorem}
\label{thm:submodular}
The tightness-aware objective function $F(P)$ is monotone and submodular. Let $P^\star$ be the optimal pivot set of size $k$, and let $P_k$ be the set returned by the greedy algorithm. Then, the greedy algorithm returns a pivot set whose objective value is guaranteed to be within a constant factor of the optimum, i.e., $F(P_k) \ge (1 - 1/e)\, F(P^\star)$.
% we have
% $F(P_k) \ge (1 - 1/e)\, F(P^\star)$.
\end{theorem}
\begin{proof}
{\color{black}{Please see the artifact repository for the detailed proof.}}
\end{proof}
% \begin{proof}[Proof Sketch]
% We formalize the above intuition by showing that the objective can be decomposed into a sum of monotone submodular functions defined on individual trajectory pairs.
% Fix a trajectory pair $(T,S) \in \mathcal{S}$ and define the normalized contribution of a pivot $p$ as follows:
% \begin{equation}
% a_p(T,S) = \frac{|\phi_p(T) - \phi_p(S)|}{D(T,S)} \ge 0
% \end{equation}
% The per-pair objective can be written as follows:
% \begin{equation}
% g_{T,S}(P) = \frac{LB_P(T,S)}{D(T,S)} = \max_{p \in P} a_p(T,S)
% \end{equation}
% The function $g_{T,S}(P)$ is monotone, since adding elements to $P$ cannot decrease a maximum. It is also submodular because it exhibits diminishing returns, i.e., for any $A \subseteq B$ and any $p \notin B$,
% \begin{equation}
% g_{T,S}(A \cup \{p\}) - g_{T,S}(A)
% \ge
% g_{T,S}(B \cup \{p\}) - g_{T,S}(B),
% \end{equation}
% as the marginal gain equals $\max(0, a_p - \max_{q\in A} a_q)$, which is no smaller than $\max(0, a_p - \max_{q\in B} a_q)$.
% Since $F(P)$ is an expectation (or empirical average) of $g_{T,S}(P)$ over sampled pairs, and nonnegative linear combinations preserve monotonicity and submodularity, $F(P)$ is also monotone submodular.
% Applying the classical result of Nemhauser et al.~\cite{nemhauser1978analysis} for greedy maximization of a monotone submodular function under a cardinality constraint yields the stated approximation guarantee.
% \end{proof}

Intuitively, tightness-aware pivot selection favors pivots that consistently separate trajectory pairs under the original distance, thereby improving the global discriminative power of the induced representation.

\subsubsection{Hard-Aware Pivot Selection}
\label{sec:hard-aware}

% While tightness-aware pivot selection optimizes the expected lower-bound quality under a uniform pair distribution, ranking performance is often dominated by a small number of hard trajectory pairs that lie close under the ground-truth distance. Trajectory similarity learning is ultimately evaluated by top-$k$ ranking metrics such as HR and Recall, where performance is primarily affected by hard trajectory pairs, i.e., pairs of trajectories with small ground-truth distances that are difficult to separate. In contrast to uniform tightness-aware selection, hard-aware selection instantiates the same tightness objective by emphasizing trajectory pairs that are most critical to the accuracy of top-$k$ retrieval. Errors in separating such near-neighbor pairs lead to ranking mistakes, whereas loose lower bounds for distant pairs typically have little impact on retrieval accuracy. Motivated by this observation, we consider a hard-aware instantiation of the tightness objective that focuses on improving lower-bound quality for ranking-critical trajectory pairs.

{\color{black}{While the tightness-aware strategy seeks pivots that improve lower-bound quality over uniformly sampled trajectory pairs, not all trajectory pairs contribute equally to retrieval performance. Trajectory similarity learning is typically evaluated using top-$k$ ranking metrics such as HR and Recall, where errors among nearby trajectories are often far more influential than errors among distant trajectories.
Consequently, optimizing the average lower-bound quality may not always yield the best retrieval performance. Motivated by this observation, we introduce a hard-aware pivot selection
strategy that focuses on ranking-critical trajectory pairs. Specifically, instead of estimating the objective using uniformly sampled pairs, we construct a hard-pair distribution consisting of near-neighbor trajectory pairs under the ground-truth distance.}}

% Motivated by this observation, we propose a hard-aware pivot selection strategy that explicitly focuses on improving lower-bound quality for ranking-critical trajectory pairs.

\textbf{Hard pair construction.}
% Let $\mathcal{T}$ denote the trajectory dataset and let $D(\cdot,\cdot)$ be the ground-truth trajectory distance. For each anchor trajectory $T \in \mathcal{T}$, we define its hard neighbor set as
To bias the tightness objective towards ranking-critical regions, we construct a hard-pair sample set as follows. Let $\mathcal{T}_{\mathit{train}}$ denote the training trajectory dataset and let $D(\cdot,\cdot)$ be the ground-truth trajectory distance. For each anchor trajectory $T \in \mathcal{T}_{\mathit{train}}$, we define its hard neighbor set as follows:
\begin{equation}
H(T) = \text{Top-}m \text{ nearest neighbors of } T \text{ under } D
\end{equation}
The hard-pair sample set is then constructed as follows:
\begin{equation}
\mathcal{S}_{\text{hard}} = \{(T,S) \mid T \in \mathcal{T}_{\mathit{train}}, S \in H(T)\}
\end{equation}
In practice, $m$ is small (e.g., 10--50), and hard neighbors are computed only for a subset of anchor trajectories to control preprocessing cost and to obtain an efficient estimate of the objective. 
% This asymmetric construction reflects the retrieval setting, where ranking quality is evaluated with respect to a query trajectory.

% \begin{equation}
% \mathcal{H}(T) = \text{Top-}m \text{ nearest neighbors of } T \text{ under } D.
% \end{equation}
% The hard-pair sample set is then constructed as
% \begin{equation}
% \mathcal{S}_{\text{hard}} = \{(T,S) \mid T \in \mathcal{T},\, S \in \mathcal{H}(T)\}.
% \end{equation}
% In practice, $m$ is small (e.g., $10$--$50$), and hard neighbors are computed only for a subset of anchor trajectories to control preprocessing cost. This asymmetric construction reflects the retrieval setting, where ranking quality is evaluated with respect to a query trajectory.

\begin{table}[t]
\centering
% \footnotesize
\small
\renewcommand{\arraystretch}{0.85}
\setlength{\tabcolsep}{3.8pt}
\caption{Trajectory dataset statistics.}
\vspace{-0.3cm}
\label{tab:dataset_properties}
\begin{tabular}{cccccc}
\toprule
Dataset & \#Total & \#Train & \#Query & \#Candidates & Avg. \#Points \\
\midrule
Geolife & 13{,}386      & 3{,}000 & 1{,}000 & 9{,}386     & 437.80 \\
Chengdu & 5{,}000       & 1{,}000 & 1{,}000 & 3{,}000     & 228.44 \\
Porto   & 1{,}601{,}579 & 3{,}000 & 500     & 1{,}598{,}079 & 48.91 \\
T-Drive & 10{,}000      & 2{,}000 & 7{,}000 & 7{,}000     & 69.99 \\
AIS     & 10{,}000      & 2{,}000 & 7{,}000 & 7{,}000     & 80.45 \\
\bottomrule
\end{tabular}
\vspace{-0.5cm}
\end{table}

% \paragraph{Hard-aware objective.}
% Using the hard-pair set $\mathcal{S}_{\text{hard}}$, we formulate the following hard-aware pivot selection problem:
% \begin{equation}
% \label{eq:hard-objective}
% F_{\text{hard}}(P) =
% \mathbb{E}_{(T,S)\sim \mathcal{S}_{\text{hard}}}
% \left[
% \frac{LB_P(T,S)}{D(T,S)}
% \right],
% \end{equation}
% By focusing the objective on near-neighbor trajectory pairs,
% hard-aware selection directly optimizes the quality of lower bounds in regions where small distortions are most likely to change the top-$k$ ranking.
\textbf{Hard-aware objective.}
Using the hard-pair set $\mathcal{S}_{\text{hard}}$, we obtain the following hard-aware instantiation of the tightness objective:
\begin{equation}
F_{\text{hard}}(P) =
\mathbb{E}_{(T,S)\sim \mathcal{S}_{\text{hard}}}
\left[
\frac{LB_P(T,S)}{D(T,S)}
\right]
\end{equation}
By concentrating the expectation on near-neighbor trajectory pairs, this instantiation emphasizes lower-bound quality in regions where small distortions are most likely to affect the top-$k$ ranking.

% The hard-aware objective $F_{\text{hard}}(P)$ shares the same functional form as Eq.~\ref{eq:tightness-objective} and differs only in the underlying pair distribution used to estimate the expectation. As a result, it remains monotone and submodular with respect to the pivot set $P$, since it is defined as an expectation of per-pair maximum functions over pivot contributions. Importantly, the optimization procedure remains unchanged, and Algorithm~\ref{alg:tightness-pivot} can be reused verbatim. The only difference is that marginal gains are estimated under the hard-aware pair distribution.
{\color{black}{\textbf{Optimization.} The hardness-aware objective differs from Eq.~\ref{eq:tightness-objective} only in the trajectory-pair distribution used to estimate the expectation. Therefore, it has the same monotonicity and submodularity properties,
and Algorithm~\ref{alg:tightness-pivot} can be reused directly. The only difference is that marginal gains are evaluated on the hard-pair set $S_{\text{hard}}$ instead of on uniformly sampled trajectory pairs.}}

% \textbf{Discussion.}
% Tightness-aware and hard-aware pivot selection optimize complementary aspects of representation quality. The former improves global approximation fidelity of lower bounds, while the latter explicitly targets ranking-critical ambiguities among near neighbors. As demonstrated in Sec.~4, hard-aware selection consistently yields further improvements in HR@$k$ and NDCG, particularly on datasets with dense local neighborhoods or highly similar trajectories.

% \textbf{Discussion.}
% Tightness-aware and hard-aware pivot selection correspond to different instantiations of the same pivot selection objective under different pair distributions. Tightness-aware selection estimates the objective under a uniform pair distribution, yielding a globally robust representation. Hard-aware selection, on the other hand, emphasizes ranking-critical near-neighbor pairs through a biased pair distribution, directly targeting the primary source of inaccuracies in top-$k$ retrieval. Both strategies are unified under the same pivot-based lower-bound framework and share an identical greedy optimization procedure. 

{\color{black}{\textbf{Discussion.} Although both strategies optimize the same lower-bound objective, they target different notions of representation quality. Tightness-aware selection aims to improve the average lower-bound quality over broadly sampled trajectory pairs, resulting in a globally robust representation. In contrast, hardness-aware selection concentrates on near-neighbor trajectory pairs that are most likely to affect retrieval rankings. Therefore, tightness-aware selection is preferable when the goal is to obtain uniformly strong lower bounds across the trajectory space, whereas hardness-aware selection is better suited for retrieval-oriented scenarios that call for distinguishing highly similar trajectories.}}

\subsection{Complexity Analysis}

We analyze the computational complexity of the proposed framework, separating offline preprocessing from online query processing. The former is performed once per dataset, while the latter is executed per query during retrieval.

\textbf{Offline preprocessing.}
% Let $|U|$ be the size of the point candidate pool, $M=|\mathcal{S}|$ the number of sampled trajectory pairs, $k$ the pivot budget, and $B$ the per-iteration candidate subset size. Let $\mathcal{T}(\mathcal{S})$ denote the set of trajectories appearing in $\mathcal{S}$ and define $L_{\mathcal{S}}=\sum_{T\in\mathcal{T}(\mathcal{S})}|T|$. We precompute pivot responses $\phi_p(T)=\min_{x\in T}\|x-p\|$ for all $p\in U$ and $T\in\mathcal{T}(\mathcal{S})$ in
% $O(|U|\cdot L_{\mathcal{S}})$ time. Given these responses, evaluating a candidate pivot on $\mathcal{S}$ costs $O(|\mathcal{S}|)$ time, and the greedy selection costs $O(kB|\mathcal{S}|)$ (worst case $B=|U|$). The use of sampled trajectory pairs allows pivot selection to scale independently of the full training set size while retaining stable selection quality in practice. Once pivots are fixed, encoding each database trajectory takes $O(k|T|)$ time (plus $O(k)$ for endpoint-to-pivot features in IPSE).
Let $|U|$ be the size of the point candidate pool, $M=|\mathcal{S}|$ be the number of sampled trajectory pairs, $k$ be the pivot budget, and $B$ be the per-iteration candidate subset size. Let $\mathcal{T}(\mathcal{S})$ denote the set of trajectories appearing in $\mathcal{S}$ and define $L_{\mathcal{S}}=\sum_{T\in\mathcal{T}(\mathcal{S})}|T|$.
For a candidate pivot $p$, its pivot responses $\phi_p(T)=\min_{x\in T}\|x-p\|$ are computed for all $T\in\mathcal{T}(\mathcal{S})$ in $O(L_{\mathcal{S}})$ time. Given these responses, evaluating the normalized marginal gain of $p$ on $\mathcal{S}$ takes $O(|\mathcal{S}|)$ time. In each iteration, evaluating $B$ candidate pivots takes $O(B(L_{\mathcal{S}}+|\mathcal{S}|))$ time, and the total greedy selection cost is $O(kB(L_{\mathcal{S}}+|\mathcal{S}|))$. 
% The use of sampled trajectory pairs allows pivot selection to scale independently of the full training set size while retaining stable selection quality (to be studied empirically).
Once pivots are fixed, encoding each database trajectory takes $O(k|T|)$ time (plus $O(k)$ for endpoint-to-pivot features in IPSE).

% \textbf{Online query processing.}
% For a query trajectory $Q$, representation construction takes
% $O(k|Q|)$ time (and $O(k)$ additional time for IPSE).
% Let $d$ be the representation dimensionality (e.g., $d=k$ for PIVOT and $d=3k$ for IPSE). Computing the representation distance to a database trajectory takes $O(d)$ time, independently of the original trajectory lengths, enabling efficient large-scale top-$k$ retrieval via standard vector search.

{\color{black}{\textbf{Online query processing.}
For a query trajectory $Q$, representation construction takes $O(k|Q|)$ time (and $O(k)$ additional time for IPSE).
Let $d$ be the representation dimensionality (e.g., $d=k$ for PIVOT and $d=3k$ for IPSE). Computing the representation distance to a database trajectory takes $O(d)$ time, independently of the original trajectory lengths. The resulting lower-bound scores are used directly for candidate ranking, and the top-$k$ trajectories are returned as the results enabling efficient large-scale top-$k$ retrieval via standard vector search.}}

\textbf{Discussion.}
In practice, both $k$ and $|\mathcal{S}|$ are small (e.g., $k=32$ and $|\mathcal{S}|$ on the order of a few thousand), and all offline preprocessing is performed once per dataset. Unlike neural embedding methods that require iterative model training, the proposed framework relies on a finite greedy selection procedure with predictable computational cost and low query-time overhead.
Moreover, from a system perspective, the proposed framework follows the same high-level retrieval pipeline as learning-based trajectory embedding methods, i.e., offline encoding followed by efficient vector distance computation at query time, while avoiding expensive model training.
% , pairwise inference, and distance-specific retraining. 
This design allows the framework to achieve strong ranking performance with predictable and low query-time cost.

\section{Experiments}
\label{sec:experiments}

\subsection{Experimental Setting}

\noindent \textbf{Datasets.}
We evaluate the proposed approach on five real-world trajectory datasets that are used in recent trajectory similarity learning studies. Specifically, we use the Geolife, Chengdu, and Porto datasets as employed in ConvTraj~\cite{chang2024revisiting}, as well as the T-Drive and AIS datasets as used in SIMformer~\cite{yang2024simformer}. These datasets cover diverse application scenarios, including pedestrian, urban taxi, and maritime vessel movement, and they exhibit varying trajectory lengths, sampling rates, and spatial distributions.
All datasets are used following the preprocessing procedures, data splits, and evaluation protocols reported in the corresponding studies. This allows us to conduct direct and consistent comparisons with previously reported methods without introducing dataset- or protocol-specific variations. Table~\ref{tab:dataset_properties} summarizes key dataset statistics.

% \noindent \textbf{Trajectory distance measures.}
% We consider three classical free-space trajectory distance measures that are commonly used in prior work: Dynamic Time Warping (DTW), Hausdorff distance, and Discrete Fr\'echet Distance (DFD). 
% % For DTW, we use the unnormalized version, consistent with existing learning-based approaches. 
% Ground-truth distances are computed offline and are used only for evaluation and for pivot selection during training.

\begin{table*}[t]
\centering
\small
\renewcommand{\arraystretch}{0.82}
\setlength{\tabcolsep}{3.2pt}
\caption{Embedding results on Geolife. Best is marked by \textbf{\underline{bold+underline}}, second best is marked by \underline{underline}. \paramark{Within each instantiation, the best and the second best results are highlighted by \colorbox{bestgray}{darker} and \colorbox{secondgray}{lighter} shading, respectively.}}
\label{tab:geolife}
\vspace{-0.2cm}
\label{tab:geolife_all}
\begin{tabular}{cc|cccc|cccc|cccc}
\toprule
\multicolumn{2}{c|}{Method}
& \multicolumn{4}{c|}{Hausdorff}
& \multicolumn{4}{c|}{DFD}
& \multicolumn{4}{c}{DTW} \\
& & HR@1 & HR@10 & HR@50 & R10@50
  & HR@1 & HR@10 & HR@50 & R10@50
  & HR@1 & HR@10 & HR@50 & R10@50 \\
\midrule

\multicolumn{2}{c|}{t2vec~\cite{li2018deep}}
& 22.00 & 24.48 & 26.64 & 44.60
& 25.70 & 28.33 & 32.44 & 53.45
& 25.30 & 28.91 & 32.81 & 55.40 \\
\multicolumn{2}{c|}{TrjSR~\cite{cao2021accurate}}
& 28.30 & 37.26 & 42.75 & 66.56
& 25.70 & 33.29 & 36.94 & 61.62
& 27.80 & 36.52 & 40.74 & 67.91 \\
\multicolumn{2}{c|}{TrajCL~\cite{chang2023contrastive}}
& 22.03 & 37.21 & 52.49 & 72.55
& 25.40 & 38.98 & 54.72 & 75.82
&  8.47 & 14.63 & 19.16 & 32.05 \\
\multicolumn{2}{c|}{NeuTraj~\cite{yao2019computing}}
& 34.53 & 48.40 & 61.88 & 80.38
& 46.37 & 64.47 & 76.43 & 94.44
& 25.13 & 33.68 & 41.72 & 62.87 \\
\multicolumn{2}{c|}{Traj2SimVec~\cite{zhang2021trajectory}}
& 26.46 & 42.39 & 48.26 & 65.12
& 27.13 & 42.75 & 50.27 & 70.20
& 16.22 & 15.39 & 20.37 & 35.09 \\
\multicolumn{2}{c|}{TrajGAT~\cite{yao2022trajgat}}
& 19.80 & 30.57 & 44.03 & 66.69
& 17.30 & 25.87 & 37.90 & 61.91
& 13.80 & 24.77 & 35.26 & 59.36 \\
\multicolumn{2}{c|}{ConvTraj~\cite{chang2024revisiting}}
& 46.17 & 63.69 & 76.12 & 95.20
& 51.80 & 68.86 & 79.52 & 97.34
& 31.70 & 46.46 & \gbest{59.26} & 83.70 \\
\midrule

\multicolumn{2}{c|}{SE}
& 38.20 & 51.25 & 57.97 & 69.20
& 55.10 & 73.92 & 86.43 & 97.91
& 37.00 & 46.89 & 55.30 & 82.38 \\
\midrule

% ======================= PIVOT (block background + global marks if hit) =======================
\multirow{6}{*}{PIVOT} & Random
& 55.20 & 69.53 & 77.07 & 95.50
& 38.90 & 49.60 & 54.54 & 89.82
& 31.20 & 39.84 & 44.64 & 76.03 \\
& Facility %~\cite{krause2008facility}
& 66.50 & 79.54 & 87.33 & 98.47
& 43.60 & 54.47 & 58.99 & 94.12
& 34.70 & 42.19 & 45.93 & 78.59 \\
& FPS %~\cite{gonzalez1985clustering}
& \secondcell{\gsecond{68.80}} & \secondcell{\gsecond{83.72}} & \secondcell{\gsecond{90.72}} & \secondcell{\gsecond{99.33}}
& \secondcell{46.80} & \secondcell{56.16} & \secondcell{60.11} & \bestcell{95.61}
& \bestcell{35.70} & \secondcell{42.64} & 46.19 & {\color{black}{\secondcell{78.82}}} \\
& KM %~\cite{kaufman1990finding}
& 67.00 & 81.43 & 88.52 & 98.81
& 46.50 & 55.27 & 59.44 & 94.32
& 34.80 & 42.29 & \secondcell{46.27} & 78.75 \\
& Tightness
& 68.00 & 82.32 & 90.11 & 98.89
& 45.80 & 54.70 & 59.17 & 94.39
& 34.50 & 42.23 & 45.76 & 78.25 \\
& Hardness
& \bestcell{\gbest{70.70}} & \bestcell{\gbest{85.89}} & \bestcell{\gbest{92.86}} & \bestcell{\gbest{99.67}}
& \bestcell{47.20} & \bestcell{56.82} & \bestcell{60.19} & \secondcell{95.42}
& \secondcell{35.50} & \bestcell{42.78} & \bestcell{46.35} & \bestcell{79.16} \\
\midrule

% ======================= PSE =======================
\multirow{6}{*}{PSE} & Random
& 47.50 & 57.79 & 60.59 & 70.89
& 73.30 & 88.61 & 94.18 & 99.56
& 42.20 & 51.77 & 57.63 & 83.97 \\
& Facility
& 48.50 & 57.94 & 60.92 & 70.94
& 74.90 & 89.71 & 95.42 & 99.71
& 42.60 & 52.00 & 57.83 & 84.18 \\
& FPS
& \secondcell{49.50} & 57.98 & 61.12 & \secondcell{71.06}
& \secondcell{\gsecond{76.50}} & \secondcell{\gsecond{90.68}} & \secondcell{\gsecond{96.23}} & \secondcell{\gsecond{99.96}}
& \bestcell{\gbest{43.20}} & \secondcell{\gsecond{52.22}} & \secondcell{58.10} & \secondcell{\gsecond{84.42}} \\
& KM %~\cite{kaufman1990finding}
& 48.40 & 58.02 & 61.11 & 71.04
& 74.60 & 90.24 & 96.01 & 99.92
& \secondcell{\gsecond{42.70}} & 52.19 & 58.08 & 84.32 \\
& Tightness
& 48.40 & \secondcell{58.09} & \secondcell{61.15} & 71.05
& 75.20 & 90.46 & 95.80 & 99.94
& 42.40 & 51.88 & 57.93 & 84.25 \\
& Hardness
& \bestcell{49.70} & \bestcell{58.23} & \bestcell{61.26} & \bestcell{71.11}
& \bestcell{\gbest{77.90}} & \bestcell{\gbest{91.30}} & \bestcell{\gbest{96.69}} & \bestcell{\gbest{99.97}}
& \bestcell{\gbest{43.20}} & \bestcell{\gbest{52.49}} & \bestcell{\gsecond{58.16}} & \bestcell{\gbest{84.52}} \\
\midrule

% ======================= IPSE =======================
\multirow{6}{*}{IPSE} & Random
& 51.50 & 66.56 & 71.88 & 89.59
& 47.00 & 59.75 & 62.90 & 92.77
& 37.80 & 46.92 & 50.77 & 79.73 \\
& Facility
& 57.80 & 71.84 & 79.60 & \secondcell{91.19}
& 54.80 & 65.63 & 68.29 & 96.34
& 40.70 & 48.79 & 52.28 & 81.73 \\
& FPS
& 58.30 & 72.22 & 80.64 & 89.87
& \secondcell{57.40} & \bestcell{68.62} & \bestcell{70.69} & \bestcell{97.63}
& \secondcell{41.10} & \bestcell{49.62} & \bestcell{53.13} & \bestcell{82.25} \\
& KM %~\cite{kaufman1990finding}
& 58.50 & 72.87 & 80.31 & 90.99
& 54.90 & 66.40 & 69.01 & 96.57
& 40.30 & 48.91 & 52.70 & \secondcell{82.04} \\
& Tightness
& \secondcell{58.70} & \secondcell{73.32} & \secondcell{81.81} & \bestcell{91.45}
& 55.80 & 65.78 & 68.77 & 96.60
& 39.70 & 48.89 & 52.50 & 81.74 \\
& Hardness
& \bestcell{59.70} & \bestcell{74.13} & \bestcell{82.53} & 90.61
& \bestcell{57.70} & \secondcell{68.43} & \secondcell{70.39} & \secondcell{97.37}
& \bestcell{41.20} & \secondcell{49.24} & \secondcell{52.96} & \bestcell{82.25} \\
\bottomrule
\end{tabular}
\vspace{-0.2cm}
\end{table*}

\begin{table*}[t]
\centering
\small
\renewcommand{\arraystretch}{0.82}
\setlength{\tabcolsep}{3.2pt}
\caption{Embedding results on Porto. See the caption of Table~\ref{tab:geolife} for notation.}
\label{tab:porto}
\vspace{-0.2cm}
\label{tab:porto_all}
\begin{tabular}{cc|cccc|cccc|cccc}
\toprule
\multicolumn{2}{c|}{Method}
& \multicolumn{4}{c|}{Hausdorff}
& \multicolumn{4}{c|}{DFD}
& \multicolumn{4}{c}{DTW} \\
& & HR@1 & HR@10 & HR@50 & R10@50
  & HR@1 & HR@10 & HR@50 & R10@50
  & HR@1 & HR@10 & HR@50 & R10@50 \\
\midrule

\multicolumn{2}{c|}{t2vec}
& 4.00 & 7.28 & 10.46 & 17.08
& 5.20 & 7.66 & 11.09 & 17.84
& 6.80 & 9.60 & 12.53 & 20.92 \\

\multicolumn{2}{c|}{TrjSR}
& 6.87 & 14.79 & 26.71 & 33.46
& 8.12 & 14.26 & 20.37 & 37.48
& 9.66 & 18.29 & 21.08 & 38.23 \\

\multicolumn{2}{c|}{TrajCL}
& 7.07 & 15.37 & 23.47 & 35.37
& 7.07 & 18.01 & 28.31 & 42.97
& 0.53 & 2.33 & 4.56 & 6.96 \\

\multicolumn{2}{c|}{NeuTraj}
& 9.30 & 24.07 & 34.22 & 51.43
& 15.13 & 33.64 & 45.47 & 66.58
& 8.40 & 16.33 & 22.98 & 34.81 \\

\multicolumn{2}{c|}{Traj2SimVec}
& 6.34 & 16.32 & 27.34 & 37.45
& 7.64 & 20.03 & 30.09 & 44.11
& 7.26 & 13.22 & 15.30 & 18.30 \\

\multicolumn{2}{c|}{TrajGAT}
& 6.48 & 18.29 & 22.31 & 43.58
& 8.39 & 20.38 & 24.31 & 39.78
& 5.73 & 12.58 & 16.97 & 20.79 \\

\multicolumn{2}{c|}{ConvTraj}
& 15.33 & 33.27 & 45.98 & 67.20
& 22.73 & 40.59 & 53.35 & 77.33
& 10.13 & 24.46 & 35.36 & 52.83 \\

\midrule
\multicolumn{2}{c|}{SE}
& 8.80 & 18.72 & 28.01 & 50.02
& 10.20 & 21.16 & 32.56 & 57.96
& 11.60 & 21.12 & 28.56 & 51.06 \\
\midrule

% ======================= PIVOT =======================
\multirow{6}{*}{PIVOT} & Random
& 15.80 & 31.46 & 43.56 & 66.98
& 14.00 & 28.24 & 36.95 & 61.28
& 13.20 & 27.82 & 34.76 & 57.68 \\

& Facility
& 21.80 & 41.76 & 55.64 & 80.84
& 19.20 & 35.52 & 44.37 & 72.24
& \secondcell{19.80} & 35.26 & 39.93 & 65.98 \\

& FPS
& \bestcell{23.80} & \secondcell{44.38} & \secondcell{\gsecond{58.52}} & \secondcell{\gsecond{82.26}}
& \bestcell{22.00} & 37.02 & \secondcell{46.39} & {73.74}
& 19.40 & 35.40 & \bestcell{{41.08}} & \secondcell{{66.74}} \\

& KMedoids
& \secondcell{23.00} & 42.92 & 57.08 & 81.52
& \secondcell{20.80} & 36.78 & 45.29 & 73.44
& \bestcell{21.00} & 35.00 & 40.48 & 65.80 \\

& Tightness
& 22.40 & 43.12 & 58.36 & 81.72
& 19.60 & \secondcell{37.32} & 46.10 & \secondcell{74.26}
& 19.40 & \secondcell{35.54} & 40.90 & 66.58 \\

& Hardness
& \bestcell{23.80} & \bestcell{\gbest{45.22}} & \bestcell{\gbest{60.38}} & \bestcell{\gbest{83.28}}
& 20.60 & \bestcell{37.72} & \bestcell{46.65} & \bestcell{74.64}
& 18.60 & \bestcell{35.94} & \secondcell{{41.06}} & \bestcell{{67.06}} \\

\midrule

% ======================= PSE =======================
\multirow{6}{*}{PSE} & Random
& 21.80 & 41.44 & 51.06 & 71.28
& 27.60 & 51.18 & 66.83 & 86.92
& 22.40 & 39.82 & 47.35 & 73.58 \\

& Facility
& \bestcell{24.00} & 42.26 & 51.33 & 71.16
& \secondcell{\gsecond{31.80}} & 52.92 & 68.88 & 87.10
& \bestcell{\gbest{24.40}} & 40.60 & \secondcell{\gsecond{48.28}} & \secondcell{\gsecond{74.26}} \\

& FPS
& 22.60 & 41.80 & 51.58 & 71.32
& 30.80 & \secondcell{\gsecond{53.00}} & 69.59 & 87.66
& \secondcell{\gsecond{23.80}} & \secondcell{\gsecond{40.62}} & 48.26 & 73.92 \\

& KMedoids
& \secondcell{23.80} & \secondcell{42.32} & \secondcell{51.68} & \secondcell{71.54}
& \bestcell{\gbest{32.20}} & \bestcell{\gbest{53.64}} & \secondcell{\gsecond{69.73}} & \bestcell{\gbest{87.92}}
& 23.40 & \bestcell{\gbest{40.80}} & 48.16 & 73.92 \\

& Tightness
& 22.60 & 41.22 & 51.24 & 71.02
& 30.00 & 52.90 & {69.48} & 87.32
& 23.20 & 40.28 & \secondcell{\gsecond{48.28}} & 74.02 \\

& Hardness
& 23.00 & \bestcell{42.50} & \bestcell{52.17} & \bestcell{71.80}
& 31.20 & \secondcell{\gsecond{53.00}} & \bestcell{\gbest{69.84}} & \secondcell{\gsecond{87.82}}
& 22.40 & 40.26 & \bestcell{\gbest{48.32}} & \bestcell{\gbest{74.44}} \\

\midrule

% ======================= IPSE =======================
\multirow{6}{*}{IPSE} & Random
& 17.80 & 35.22 & 46.22 & 71.20
& 17.80 & 33.48 & 43.22 & 68.94
& 17.00 & 33.08 & 40.70 & 65.02 \\

& Facility
& 22.80 & 43.20 & 54.22 & 77.98
& 22.80 & 40.40 & 51.16 & 77.34
& \bestcell{23.60} & 39.06 & 44.70 & 70.16 \\

& FPS
& 24.20 & \secondcell{44.24} & \secondcell{56.25} & {77.42}
& \secondcell{25.20} & {42.30} & \secondcell{53.98} & {78.58}
& 22.20 & \secondcell{39.68} & \bestcell{46.04} & \bestcell{71.56} \\

& KMedoids
& \bestcell{\gbest{24.80}} & 43.82 & 55.40 & \secondcell{78.36}
& \bestcell{25.40} & 41.40 & 52.51 & 78.26
& \secondcell{23.40} & 38.60 & 45.43 & 70.18 \\

& Tightness
& 21.80 & 43.76 & 56.02 & 77.44
& 23.20 & \secondcell{42.50} & 53.88 & \secondcell{79.14}
& 22.00 & 39.60 & \secondcell{45.91} & 71.00 \\

& Hardness
& \secondcell{\gsecond{24.60}} & \bestcell{\gsecond{44.80}} & \bestcell{58.04} & \bestcell{79.26}
& {24.60} & \bestcell{42.60} & \bestcell{54.08} & \bestcell{79.18}
& 21.20 & \bestcell{39.74} & \bestcell{46.04} & \secondcell{71.06} \\

\bottomrule
\end{tabular}
\vspace{-0.4cm}
\end{table*}

\begin{table*}[t]
\centering
\small
\renewcommand{\arraystretch}{0.82}
\setlength{\tabcolsep}{4.2pt}
\caption{Embedding results on Chengdu. See the caption of Table~\ref{tab:geolife} for notation.}
\label{tab:chengdu}
\vspace{-0.2cm}
\label{tab:chengdu_all}
\begin{tabular}{cc|ccc|ccc|ccc}
\toprule
\multicolumn{2}{c|}{Method}
& \multicolumn{3}{c|}{Hausdorff}
& \multicolumn{3}{c|}{DFD}
& \multicolumn{3}{c}{DTW} \\
& & HR@5 & HR@10 & R10@50
  & HR@5 & HR@10 & R10@50
  & HR@5 & HR@10 & R10@50 \\
\midrule
\multicolumn{2}{c|}{t2vec~\cite{li2018deep}}
& 16.16 & 16.52 & 30.69
& 21.34 & 22.34 & 44.06
& 21.00 & 22.68 & 46.03 \\
\multicolumn{2}{c|}{TrjSR~\cite{cao2021accurate}}
& 10.44 & 12.09 & 25.44
& 11.08 & 12.77 & 25.96
& 16.42 & 18.32 & 36.63 \\
\multicolumn{2}{c|}{TrajCL~\cite{chang2023contrastive}}
& 22.14 & 27.04 & 61.93
& 19.00 & 21.42 & 51.79
& 23.48 & 27.20 & 59.89 \\
\multicolumn{2}{c|}{NeuTraj~\cite{yao2019computing}}
& 21.82 & 23.27 & 39.09
& 32.28 & 34.70 & 49.51
& 28.40 & 29.33 & 46.90 \\
\multicolumn{2}{c|}{Traj2SimVec~\cite{zhang2021trajectory}}
& 18.34 & 20.17 & 43.67
& 25.16 & 28.33 & 42.78
& 18.66 & 19.37 & 40.69 \\
\multicolumn{2}{c|}{TrajGAT~\cite{yao2022trajgat}}
& 19.98 & 25.26 & 57.57
& 16.24 & 19.11 & 49.18
& 23.96 & 28.32 & 65.41 \\
\multicolumn{2}{c|}{ConvTraj~\cite{chang2024revisiting}}
& 36.26 & 42.78 & 76.67
& 53.34 & 58.18 & 93.14
& 34.90 & 40.40 & 76.23 \\
\midrule

\multicolumn{2}{c|}{SE}
& 33.18 & 36.79 & 53.04
& 58.20 & 64.78 & 97.46
& 34.06 & 36.78 & 72.40 \\
\midrule

% ======================= PIVOT (block-wise background) =======================
\multirow{6}{*}{PIVOT} & Random
& 59.02 & 63.43 & 95.56
& 33.86 & 36.35 & 83.29
& 31.66 & 32.72 & 68.85 \\
& Facility
& 50.00 & 54.70 & 92.92
& 30.18 & 33.54 & 79.06
& 30.88 & 32.30 & 67.91 \\
& FPS
& 73.42 & 78.31 & 99.81
& 39.04 & 41.25 & 92.11
& \bestcell{33.04} & \bestcell{34.26} & \bestcell{70.74} \\
& KM
& 50.16 & 55.53 & 93.40
& 30.28 & 33.39 & 79.45
& 29.62 & 31.10 & 66.53 \\
& Tightness
& \secondcell{\underline{80.78}} & \bestcell{\textbf{\underline{85.63}}} & \bestcell{\textbf{\underline{99.95}}}
& \secondcell{40.46} & \secondcell{42.89} & \secondcell{94.39}
& \secondcell{32.44} & \secondcell{34.14} & 70.17 \\
& Hardness
& \bestcell{\textbf{\underline{80.88}}} & \secondcell{\underline{85.54}} & \secondcell{\underline{99.94}}
& \bestcell{40.70} & \bestcell{42.92} & \bestcell{94.45}
& 32.42 & 33.90 & \secondcell{70.28} \\
\midrule

% ======================= PSE (block-wise background) =======================
\multirow{6}{*}{PSE} & Random
& 41.70 & 43.17 & 55.11
& 83.48 & 86.16 & 99.52
& 42.10 & 43.84 & 77.12 \\
& Facility
& 40.78 & 42.83 & 55.04
& 80.76 & 84.37 & 99.53
& 41.48 & 43.68 & 76.94 \\
& FPS
& 42.38 & 43.82 & 55.55
& 87.66 & 90.07 & \bestcell{\textbf{\underline{100.00}}}
& \secondcell{\underline{42.74}} & 44.80 & 78.25 \\
& KM
& 40.52 & 42.68 & 55.01
& 79.66 & 83.08 & \secondcell{\underline{99.58}}
& 40.82 & 43.01 & 76.96 \\
& Tightness
& \secondcell{42.78} & \bestcell{44.15} & \bestcell{55.61}
& \secondcell{\underline{88.46}} & \secondcell{\underline{90.97}} & \bestcell{\textbf{\underline{100.00}}}
& 42.70 & \bestcell{\textbf{\underline{44.96}}} & \bestcell{\textbf{\underline{78.42}}} \\
& Hardness
& \bestcell{42.84} & \secondcell{43.99} & \secondcell{55.56}
& \bestcell{\textbf{\underline{88.82}}} & \bestcell{\textbf{\underline{91.24}}} & \bestcell{\textbf{\underline{100.00}}}
& \bestcell{\textbf{\underline{42.76}}} & \secondcell{\underline{44.87}} & \secondcell{\underline{78.37}} \\
\midrule

% ======================= IPSE (block-wise background) =======================
\multirow{6}{*}{IPSE} & Random
& 58.78 & 63.26 & 95.56
& 34.52 & 36.53 & 83.29
& 32.26 & 33.00 & 68.85 \\
& Facility
& 49.62 & 54.58 & 92.92
& 30.88 & 33.74 & 79.06
& 31.54 & 32.53 & 67.91 \\
& FPS
& 72.92 & 78.17 & 99.81
& 39.46 & 41.38 & 92.11
& \bestcell{33.50} & \bestcell{34.37} & \bestcell{70.74} \\
& KM
& 50.28 & 55.44 & 93.40
& 31.24 & 33.86 & 79.45
& 30.60 & 31.57 & 66.53 \\
& Tightness
& \secondcell{80.28} & \bestcell{85.47} & \bestcell{\textbf{\underline{99.95}}}
& \secondcell{40.90} & \secondcell{42.99} & \secondcell{94.39}
& \secondcell{32.92} & \secondcell{34.22} & 70.17 \\
& Hardness
& \bestcell{80.38} & \secondcell{85.40} & \secondcell{\underline{99.94}}
& \bestcell{41.08} & \bestcell{43.01} & \bestcell{94.45}
& 32.86 & 33.98 & \secondcell{70.28} \\
\bottomrule
\end{tabular}
\vspace{-0.2cm}
\end{table*}

\begin{table*}[t]
\centering
\small
\renewcommand{\arraystretch}{0.82}
\setlength{\tabcolsep}{3.2pt}
\caption{Embedding results on T-Drive. See the caption of Table~\ref{tab:geolife} for notation.}
\label{tab:tdrive}
\vspace{-0.2cm}
\label{tab:tdrive_all}
\begin{tabular}{cc|cccc|cccc|cccc}
\toprule
\multicolumn{2}{c|}{Method}
& \multicolumn{4}{c|}{Hausdorff}
& \multicolumn{4}{c|}{DFD}
& \multicolumn{4}{c}{DTW} \\
& & HR@1 & HR@10 & HR@50 & R10@50
  & HR@1 & HR@10 & HR@50 & R10@50
  & HR@1 & HR@10 & HR@50 & R10@50 \\
\midrule

\multicolumn{2}{c|}{NeuTraj}
& 13.34 & 31.04 & 48.57 & 68.57
& 26.30 & 44.94 & 57.10 & 83.40
& 10.18 & 22.71 & 35.64 & 49.31 \\
\multicolumn{2}{c|}{Traj2SimVec}
& 14.59 & 32.05 & 48.88 & 67.67
& 24.94 & 43.00 & 55.30 & 80.54
& 11.88 & 26.59 & 41.17 & 58.01 \\
\multicolumn{2}{c|}{T3S}
& 19.35 & 41.10 & 57.01 & 82.14
& 31.74 & 49.80 & 61.32 & 88.84
& 9.19 & 22.58 & 36.91 & 50.26 \\
\multicolumn{2}{c|}{TMN-NM}
& 21.92 & 44.42 & 60.41 & 86.23
& 32.43 & 50.85 & 62.14 & 89.88
& 9.83 & 23.34 & 37.57 & 52.48 \\

\multicolumn{2}{c|}{SIMFormer w/ Euc.}
& 36.97 & 57.08 & 68.31 & 95.43
& 34.52 & 52.50 & 63.35 & 91.26
& 11.05 & 24.36 & 38.33 & 52.89 \\
\multicolumn{2}{c|}{SIMFormer w/ Cos.}
& 34.40 & 53.46 & 64.15 & 92.84
& 28.47 & 45.19 & 55.47 & 84.68
& \textbf{\underline{19.88}} & \textbf{\underline{36.00}} & \textbf{\underline{52.56}} & \textbf{\underline{73.94}} \\
\multicolumn{2}{c|}{SIMFormer w/ Cheby.}
& 43.66 & 66.79 & 78.85 & 98.04
& 41.04 & 61.35 & 73.22 & 94.41
& 14.50 & \underline{29.16} & \underline{43.72} & \underline{58.80} \\

\midrule

\multicolumn{2}{c|}{SE}
& 5.87 & 12.63 & 20.32 & 25.83
& 12.67 & 26.38 & 40.26 & 70.66
& 8.04 & 16.41 & 24.48 & 42.16 \\

\midrule

% ======================= PIVOT (block-wise background) =======================
\multirow{6}{*}{PIVOT} & Random
& 31.81 & 47.41 & 57.11 & 86.26
& 9.97 & 18.14 & 26.82 & 53.69
& 9.91 & 17.36 & 25.93 & 43.34 \\
& Facility
& 49.76 & 68.32 & 78.64 & 96.38
& 13.47 & 22.59 & 32.20 & 66.56
& \secondcell{11.59} & \bestcell{19.88} & \bestcell{28.34} & \bestcell{47.62} \\
& FPS
& 52.56 & 71.68 & 82.70 & 99.32
& \secondcell{13.73} & \bestcell{23.01} & 32.93 & \secondcell{67.85}
& 11.57 & 19.76 & 28.08 & 47.39 \\
& KM
& 48.36 & 66.78 & 77.17 & 96.09
& 12.99 & 22.39 & 31.99 & 65.87
& 11.21 & \secondcell{19.84} & \secondcell{28.33} & 47.48 \\
& Tightness
& \secondcell{\underline{52.63}} & \secondcell{\underline{72.69}} & \bestcell{\textbf{\underline{83.82}}} & \bestcell{\textbf{\underline{99.57}}}
& \bestcell{13.76} & 22.78 & \secondcell{32.95} & 67.83
& \bestcell{11.76} & 19.77 & 28.17 & \secondcell{47.57} \\
& Hardness
& \bestcell{\textbf{\underline{53.44}}} & \bestcell{\textbf{\underline{72.78}}} & \secondcell{\underline{83.68}} & \secondcell{\underline{99.44}}
& 13.61 & \secondcell{22.96} & \bestcell{33.01} & \bestcell{68.15}
& 11.54 & 19.75 & 28.19 & 47.52 \\

\midrule

% ======================= PSE (block-wise background) =======================
\multirow{6}{*}{PSE} & Random
& 14.40 & 23.04 & 31.72 & 35.57
& 44.63 & 61.56 & 71.90 & 93.97
& 17.47 & 27.33 & 34.19 & 54.84 \\
& Facility
& 16.37 & 24.84 & 34.07 & 37.74
& \secondcell{\underline{51.83}} & 69.42 & 79.44 & 97.42
& \bestcell{\underline{18.97}} & \bestcell{28.72} & 35.51 & 56.68 \\
& FPS
& 16.46 & \secondcell{25.21} & 34.58 & 38.14
& 51.36 & 69.60 & 80.35 & \bestcell{\textbf{\underline{98.02}}}
& 18.74 & 28.63 & 35.54 & 56.82 \\
& KM
& 16.13 & 24.73 & 33.93 & 37.63
& 51.34 & 68.89 & 79.01 & 97.38
& 18.73 & 28.60 & 35.46 & 56.62 \\
& Tightness
& \bestcell{16.76} & \bestcell{25.30} & \bestcell{34.70} & \bestcell{38.25}
& 51.39 & \secondcell{\underline{69.71}} & \secondcell{\underline{80.66}} & \secondcell{\underline{97.99}}
& 18.67 & 28.66 & \secondcell{35.63} & \secondcell{56.88} \\
& Hardness
& \secondcell{16.60} & 25.19 & \secondcell{34.60} & \secondcell{38.19}
& \bestcell{\textbf{\underline{52.04}}} & \bestcell{\textbf{\underline{70.06}}} & \bestcell{\textbf{\underline{80.79}}} & \bestcell{\textbf{\underline{98.02}}}
& \secondcell{18.80} & \secondcell{28.70} & \bestcell{35.66} & \bestcell{56.91} \\

\midrule

% ======================= IPSE (block-wise background) =======================
\multirow{6}{*}{IPSE} & Random
& \bestcell{20.29} & \bestcell{33.92} & \bestcell{45.32} & \bestcell{50.68}
& 29.54 & 41.93 & 49.63 & 83.38
& \bestcell{17.69} & \bestcell{26.80} & \bestcell{33.61} & \bestcell{54.17} \\
& Facility
& 12.61 & 20.20 & 27.85 & 33.05
& \secondcell{38.07} & \secondcell{52.96} & \secondcell{60.68} & 87.24
& 16.07 & 24.79 & 31.13 & 50.46 \\
& FPS
& 11.29 & 19.12 & 26.44 & 31.89
& 34.06 & 48.87 & 56.89 & 83.62
& 14.80 & 23.70 & 29.95 & 48.82 \\
& KM
& \secondcell{13.43} & \secondcell{21.26} & \secondcell{29.24} & \secondcell{34.26}
& \bestcell{39.40} & \bestcell{55.18} & \bestcell{63.34} & \bestcell{90.69}
& \secondcell{16.53} & \secondcell{25.62} & \secondcell{32.11} & \secondcell{52.09} \\
& Tightness
& 12.19 & 20.20 & 27.76 & 33.02
& 36.19 & 52.15 & 60.34 & \secondcell{87.73}
& 15.39 & 24.41 & 30.60 & 49.80 \\
& Hardness
& 11.99 & 20.16 & 27.70 & 32.94
& 36.39 & 52.25 & 60.36 & 87.71
& 15.74 & 24.71 & 31.10 & 50.60 \\

\bottomrule
\end{tabular}
\vspace{-0.2cm}
\end{table*}

\begin{table*}[t]
\centering
\small
\renewcommand{\arraystretch}{0.82}
\setlength{\tabcolsep}{3.2pt}        % 列间距（你现在 4.0pt），可降到 2.8~3.5pt
\caption{Embedding results on AIS. See the caption of Table~\ref{tab:geolife} for notation.}
\label{tab:ais}
\vspace{-0.2cm}
\label{tab:ais_all}
\begin{tabular}{cc|cccc|cccc|cccc}
\toprule
\multicolumn{2}{c|}{Method}
& \multicolumn{4}{c|}{Hausdorff}
& \multicolumn{4}{c|}{DFD}
& \multicolumn{4}{c}{DTW} \\
& & HR@1 & HR@10 & HR@50 & R10@50
  & HR@1 & HR@10 & HR@50 & R10@50
  & HR@1 & HR@10 & HR@50 & R10@50 \\
\midrule

% ===== baselines (UPDATED to match your Table 3) =====
\multicolumn{2}{c|}{NeuTraj}
& 15.14 & 32.29 & 49.99 & 66.14
& 29.01 & 51.26 & 65.81 & 85.57
& 16.18 & 33.77 & 49.26 & 69.12 \\
\multicolumn{2}{c|}{Traj2SimVec}
& 22.25 & 40.38 & 57.26 & 72.40
& 36.43 & 59.16 & 71.50 & 91.26
& 18.13 & 36.04 & 50.46 & 71.84 \\
\multicolumn{2}{c|}{T3S}
& 11.49 & 30.12 & 51.54 & 67.10
& 32.92 & 56.64 & 71.03 & 89.59
& 14.87 & 32.98 & 50.05 & 69.01 \\
\multicolumn{2}{c|}{TMN-NM}
& 24.65 & 47.99 & 66.48 & 84.96
& 37.58 & 60.13 & 72.74 & 92.72
& 15.41 & 33.26 & 49.51 & 68.61 \\
\multicolumn{2}{c|}{SIMFormer w/ Euc.}
& 31.66 & 56.11 & 72.32 & 91.72
& 37.30 & 59.69 & 72.51 & 92.49
& 15.89 & 33.08 & 48.83 & 68.08 \\
\multicolumn{2}{c|}{SIMFormer w/ Cos.}
& 24.03 & 45.31 & 63.13 & 82.11
& 27.73 & 48.70 & 64.17 & 85.21
& 21.73 & 42.15 & 60.05 & 78.16 \\
\multicolumn{2}{c|}{SIMFormer w/ Cheby.}
& 35.58 & 59.68 & 75.88 & 92.76
& 39.52 & 62.83 & 76.75 & 93.26
& 17.66 & 35.58 & 51.26 & 70.48 \\

\midrule

% ===== SE (from AIS excel) =====
\multicolumn{2}{c|}{SE}
& 15.30 & 25.40 & 32.19 & 42.56
& 26.71 & 43.90 & 56.26 & 73.25
& 17.16 & 33.06 & 45.08 & 67.40 \\

\midrule

% ======================= PIVOT (block-wise background, column-wise inside block) =======================
\multirow{6}{*}{PIVOT} & Random
& 35.83 & 54.48 & 66.95 & 91.15
& 20.30 & 33.70 & 41.75 & 75.98
& 19.01 & 32.75 & 41.66 & 70.69 \\
& Facility
& 52.39 & 70.79 & 82.87 & 97.98
& 26.63 & 39.38 & 47.07 & 85.40
& 21.86 & 36.07 & 44.80 & 76.39 \\
& FPS
& \secondcell{\underline{57.96}} & 76.27 & 86.74 & 98.89
& \secondcell{28.97} & {41.50} & 48.30 & 86.96
& \bestcell{22.63} & 36.64 & 45.32 & 77.07 \\
& KM
& 50.76 & 70.43 & 82.33 & 97.83
& 26.37 & 39.26 & 46.98 & 85.28
& 21.43 & 35.99 & 44.81 & 76.39 \\
& Tightness
& 57.91 & \secondcell{\underline{77.29}} & \secondcell{\underline{88.17}} & \secondcell{\underline{99.14}}
& 28.71 & \secondcell{41.74} & \bestcell{48.71} & \secondcell{87.36}
& \secondcell{22.50} & \bestcell{36.94} & \bestcell{45.54} & \bestcell{77.29} \\
& Hardness
& \bestcell{\textbf{\underline{59.59}}} & \bestcell{\textbf{\underline{78.12}}} & \bestcell{\textbf{\underline{88.77}}} & \bestcell{\textbf{\underline{99.33}}}
& \bestcell{29.20} & \bestcell{41.89} & \secondcell{48.58} & \bestcell{87.44}
& 22.49 & \secondcell{36.71} & \secondcell{45.46} & \secondcell{77.24} \\

\midrule

% ======================= PSE (block-wise background, column-wise inside block) =======================
\multirow{6}{*}{PSE} & Random
& 34.81 & 45.76 & 49.75 & 59.94
& 53.10 & 73.32 & 83.79 & 96.79
& 33.83 & 52.09 & 62.20 & 87.41 \\
& Facility
& 40.10 & 49.07 & 52.43 & 61.13
& {61.11} & 80.16 & 89.76 & 98.45
& 35.61 & \secondcell{\underline{53.92}} & 63.24 & \secondcell{\underline{88.52}} \\
& FPS
& \bestcell{40.97} & \bestcell{49.45} & \bestcell{52.76} & \bestcell{61.30}
& \bestcell{\textbf{\underline{62.09}}} & \bestcell{\textbf{\underline{80.74}}} & \secondcell{\underline{90.30}} & \secondcell{\underline{98.50}}
& \bestcell{\textbf{\underline{36.04}}} & \bestcell{\textbf{\underline{53.96}}} & \bestcell{\textbf{\underline{63.30}}} & 88.51 \\
& KM
& 39.70 & 48.86 & 52.35 & 61.10
& 60.37 & 79.51 & 89.37 & 98.40
& 35.54 & 53.68 & 63.20 & 88.45 \\
& Tightness
& 40.21 & \secondcell{49.30} & 52.63 & \secondcell{61.25}
& \secondcell{\underline{61.71}} & {80.64} & 90.10 & \secondcell{\underline{98.50}}
& \secondcell{\underline{35.91}} & 53.86 & \secondcell{\underline{63.28}} & \bestcell{\textbf{\underline{88.55}}} \\
& Hardness
& \secondcell{40.73} & 49.21 & \secondcell{52.73} & 61.24
& 61.59 & \secondcell{\underline{80.65}} & \bestcell{\textbf{\underline{90.38}}} & \bestcell{\textbf{\underline{98.51}}}
& 35.56 & 53.75 & \secondcell{\underline{63.28}} & 88.50 \\

\midrule

% ======================= IPSE (block-wise background, column-wise inside block) =======================
\multirow{6}{*}{IPSE} & Random
& 32.61 & 45.91 & \bestcell{52.57} & \bestcell{63.56}
& 42.40 & 60.99 & 68.27 & 93.49
& 32.46 & 49.29 & 58.66 & 84.74 \\
& Facility
& \bestcell{37.33} & \secondcell{47.56} & 50.95 & 60.35
& \secondcell{53.99} & \bestcell{73.04} & 81.45 & 97.26
& \secondcell{33.97} & \bestcell{52.46} & \secondcell{61.56} & \secondcell{87.37} \\
& FPS
& 36.70 & 46.27 & 49.82 & 59.75
& 53.19 & 72.23 & 81.71 & \secondcell{97.44}
& 33.76 & 51.73 & 61.24 & 86.94 \\
& KM
& \bestcell{37.33} & \bestcell{47.59} & \secondcell{51.17} & \secondcell{60.56}
& \bestcell{54.06} & 72.73 & 81.20 & 97.34
& \bestcell{34.01} & \secondcell{52.34} & \bestcell{61.58} & \bestcell{87.42} \\
& Tightness
& 36.77 & 46.86 & 50.30 & 60.19
& 53.63 & \secondcell{72.89} & \secondcell{82.00} & \bestcell{97.58}
& 33.90 & 51.81 & 61.29 & 87.01 \\
& Hardness
& \secondcell{37.07} & 46.74 & 50.26 & 59.93
& 53.03 & 72.87 & \bestcell{82.13} & \bestcell{97.58}
& 33.70 & 51.58 & 61.14 & 86.92 \\

\bottomrule
\end{tabular}
\vspace{-0.3cm}
\end{table*}

\noindent \textbf{Baselines.}
We compare the proposed lower-bound representations with state-of-the-art learning-based trajectory similarity methods covered in the ConvTraj~\cite{chang2024revisiting} and SIMformer~\cite{yang2024simformer} studies, including t2vec~\cite{li2018deep}, TrjSR~\cite{cao2021accurate}, NeuTraj~\cite{yao2019computing}, Traj2SimVec~\cite{zhang2021trajectory}, TrajGAT~\cite{yao2022trajgat}, TrajCL~\cite{chang2023contrastive}, ConvTraj~\cite{chang2024revisiting}, T3S~\cite{yang2021t3s}, TMN-NM~\cite{yang2022tmn}, and SIMformer~\cite{yang2024simformer}. For these baselines, we report the results from the ConvTraj and SIMformer studies directly under their reported experimental settings. 
{\color{black}{Since the experiments are conducted on the same preprocessed datasets, data splits, and evaluation protocols as those studies, the reported results remain directly comparable.}}
This is common practice in recent trajectory similarity learning studies~\cite{chang2024revisiting} and avoids the re-training of complex neural models with potentially different hyperparameter configurations, which may lead to unfair comparisons.

In addition to the above learning-based baselines, we compare different pivot selection methods for constructing lower-bound representations. Specifically, we consider several widely used heuristic pivot selection methods~\cite{zhu2022pivot,chen2015efficient}, including random selection, facility location based selection~\cite{krause2008facility}, farthest point sampling (FPS)~\cite{gonzalez1985clustering}, and $k$-medoids clustering~\cite{kaufman1990finding}. These methods have been used for metric indexing and similarity search to improve pivot coverage or dispersion, but they do not explicitly optimize ranking quality. Our proposed tightness- and hard-aware pivot selection methods serve as learning-based alternatives under the same framework.

\noindent \textbf{Evaluation protocol and metrics.}
% Following prior work, we evaluate retrieval effectiveness using ranking-based metrics, i.e., Hit Ratio at $k$ (HR@$k$). For each query trajectory, candidate trajectories are ranked according to the representation distance, and the resulting ranking is compared against the ground-truth ranking induced by the original trajectory distance. Since the primary goal of trajectory similarity learning is efficient retrieval rather than exact distance reconstruction, we focus exclusively on ranking quality.
{\color{black}{Following the standard evaluation protocol adopted in trajectory similarity learning~\cite{chang2024revisiting}, }}
we evaluate retrieval effectiveness using ranking-based metrics, including Hit Ratio at $k$ (HR@$k$) and Recall@50 with top-10 ground-truth neighbors (R10@50)~\cite{deng2024learning}. For each query trajectory, candidate trajectories are ranked according to the representation distance and are compared against the ground-truth ranking induced by the original trajectory distance. 
% Since our goal is efficient trajectory retrieval, we focus exclusively on ranking quality.
Ground-truth similarity rankings are induced by the considered trajectory distance measure (i.e., Hausdorff distance, DFD, or DTW), and they are used only for evaluation and for pivot selection during training. 
% {\color{black}{This follows the standard evaluation protocol adopted by existing trajectory similarity learning methods~\cite{chang2024revisiting}.}}

\noindent \textbf{Implementation details.}
% Pivot selection is performed offline using sampled trajectory pairs from the training set. Once the pivot set is fixed, all database trajectories are encoded into fixed-length vectors and stored. At query time, only the query trajectory is encoded, and similarity computation is conducted using the corresponding lower-bound–based representation distance. All experiments are conducted under the same dataset settings and evaluation protocols as the compared methods.
Pivot selection is performed offline using sampled trajectory pairs from the training set. Unless stated otherwise, the pivot budget and the number of sampled trajectory pairs are set to $k=32$ and $|\mathcal{S}| = 512$, respectively. For hard-aware pivot selection, hard pairs are constructed by selecting the top-$m$ nearest neighbors under the ground-truth distance for each anchor trajectory, with $m=50$.
% Once the pivot set is fixed, all database trajectories are encoded into fixed-length vectors and stored. At query time, only the query trajectory is encoded, and similarity computation is conducted using the corresponding lower-bound--based representation distance.

\subsection{Effectiveness}
\label{sec:effectiveness}

% \subsection{Overall Performance}
\noindent \underline{\textbf{Overall Performance.}} 
% We evaluate the overall ranking performance of the proposed lower-bound representations against state-of-the-art learning-based trajectory embeddings. Results on five real-world datasets are reported in Tables~2--6. Overall, our methods consistently outperform neural baselines across most dataset--distance combinations. In particular, on Hausdorff and DFD distances, the proposed representations improve HR@1 by approximately 20\%--60\% over the strongest learning-based baselines, while achieving comparable or better gains on HR@10 and HR@50. 
% Under DTW, where several neural methods suffer from severe performance degradation, our approach remains stable and effective. Across datasets, PSE improves HR@50 by roughly 15\%--40\% compared to neural baselines in most settings. On the large-scale Porto dataset, despite the presence of over 1.5 million trajectories, our methods continue to deliver strong ranking performance, with consistent improvements of more than 30\% on key metrics such as HR@1 and R10@50. These results demonstrate that lower-bound representations provide a robust and competitive alternative to learning-based trajectory embeddings across diverse distance measures and data scales.
% Notably, even the extremely simple SE baseline achieves surprisingly strong performance in several settings, indicating that coarse geometric lower-bound cues already capture a substantial portion of trajectory similarity.
We evaluate the overall ranking performance of the proposed lower-bound representations against state-of-the-art trajectory embeddings. Results on the five datasets are reported in Tables~\ref{tab:geolife}--\ref{tab:ais}. Overall, our methods outperform the baselines across most dataset--distance combinations. In particular, on the Hausdorff and DFD distances, the proposed representations improve HR@1 by approximately 20\%--60\% over the strongest learning-based baselines, while achieving better gains on HR@10 and HR@50. Under DTW, where several baseline methods suffer from severe performance degradation, our approach remains stable and effective.
Across the datasets, PSE improves HR@50 by roughly 15\%--40\% compared to the baselines in most settings. On the large-scale Porto dataset, our methods continue to deliver strong ranking performance, achieving improvements of 30+\% on key metrics such as HR@1 and R10@50. These results demonstrate that lower-bound representations provide a robust and competitive alternative to learning-based trajectory embeddings across diverse distance measures and data volumes. Notably, even the extremely simple SE baseline achieves surprisingly strong performance in several settings, indicating that coarse geometric lower-bound cues already capture a substantial portion of trajectory similarity.
Furthermore, we observe that the proposed lower-bound representations achieve competitive performance across all three trajectory distance measures, despite their substantially different matching semantics.
% DTW focuses on alignment-based matching, DFD emphasizes the maximum deviation along an optimal traversal, and Hausdorff measures the worst-case point-wise discrepancy between trajectories.
{\color{black}{An additional observation is that the proposed framework separates representation construction from distance evaluation.
Distance-independent pivot selection strategies such as Random and FPS remain competitive across DTW, DFD, and Hausdorff distance, suggesting that the learned representations capture geometric characteristics useful across multiple trajectory similarity definitions.
Distance-aware pivot selection can further improve retrieval quality when a target distance is known in advance.}}

\noindent \underline{\textbf{Ablation Study.}}
We further analyze the impact of the different representation instantiations and pivot selection strategies. Comparing SE, PIVOT, PSE, and IPSE in Tables~\ref{tab:geolife}--\ref{tab:ais}, we observe clear distance-dependent behaviors. PIVOT consistently achieves the best performance under Hausdorff distance, yielding improvements of approximately 10\%--25\% over PSE and IPSE. 
This behavior can be attributed to the max--min structure of Hausdorff distance, which is naturally aligned with pivot-based lower bounds.
{\color{black}{In contrast, the endpoint-based component does not constitute a valid lower bound for Hausdorff distance. As a result, instantiations that incorporate endpoint information (i.e., SE, PSE, and IPSE) tend to exhibit weaker performance, confirming the importance of respecting distance-specific lower-bound validity.}}
% Moreover, for alignment-based distances such as DFD and DTW, PSE performs best in most cases, improving HR@10 and HR@50 by around 5\%--15\% over PIVOT, as the combination of endpoint and pivot-based lower bounds preserves alignment-sensitive rankings better. IPSE achieves comparable but slightly less stable performance in some settings.
{\color{black}{Next, for alignment-based distances such as DFD and DTW, PSE consistently outperforms both SE and PIVOT across all datasets and evaluation metrics, improving HR@10 and HR@50 by around 5\%--15\% over PIVOT. This behavior is consistent with the construction of PSE. Since both SE and PIVOT constitute valid lower bounds for DFD and DTW, the max-combination used in PSE produces a lower bound that is never weaker than either individual component. Consequently, PSE is able to preserve trajectory similarity relationships more effectively than is either SE or PIVOT alone. IPSE achieves comparable but slightly less stable performance in some settings.}}
We also evaluate different pivot selection strategies. Learned pivot selection consistently outperforms the random and heuristic baselines, with the tightness- and hardness-aware strategies providing additional gains of 5\%--20\% on the top-$k$ metrics in most cases. Notably, hardness-aware selection is particularly effective for HR@1 and HR@10, confirming the importance of focusing on hard, near-neighbor trajectory pairs during pivot learning.

% \noindent \underline{\textbf{Ranking Quality.}} Number of Inversions / Spearman Correlations

\begin{figure}[t]
    \centering
    \includegraphics[width=\columnwidth]{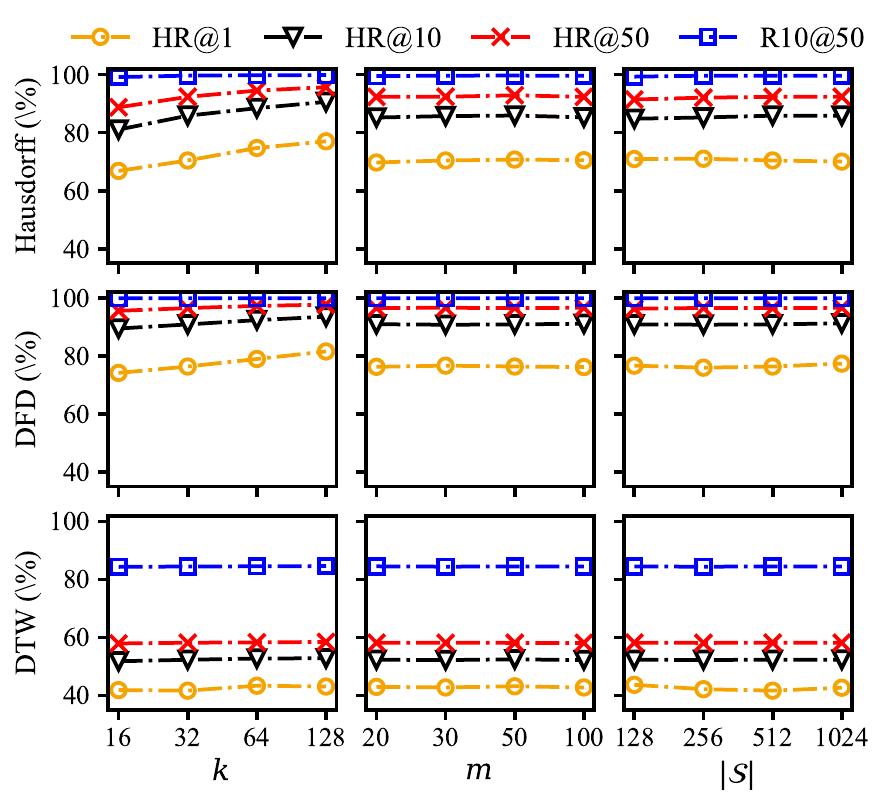}
    \vspace{-0.7cm}
    \caption{Hyper-parameter study on Geolife.}
    \label{fig:sensitivity-geolife}
    \vspace{-0.5cm}
\end{figure}

\subsection{Sensitivity Study}
We study the sensitivity of our framework to key hyper-parameters on the Geolife dataset, with results shown in Fig.~\ref{fig:sensitivity-geolife}. Following the best-performing configuration in Section~\ref{sec:effectiveness}, we use \textsc{PIVOT} for Hausdorff distance and \textsc{PSE} for DFD/DTW, and we adopt the hard-aware pivot selection strategy throughout. We vary one parameter at a time while keeping the others fixed (default \(k=32\), \(m=50\), and \(|\mathcal{S}|=512\)).

\noindent \textbf{Effect of pivot budget \(k\).} Increasing \(k\) consistently improves ranking quality for all three distances. The improvement is most pronounced for the Hausdorff distance, where larger pivot sets provide better spatial coverage and thus stronger discrimination. For DFD and DTW, the gains are still monotonic but saturate quickly. The performance becomes relatively stable once \(k\) reaches a moderate value (e.g., \(k=32\) or \(64\)), indicating that a small set of informative pivots captures most ranking-relevant information.

\noindent \textbf{Effect of hardness parameter \(m\).} Varying \(m\) (the number of near neighbors used to construct hard pairs) has a noticeably smaller impact than \(k\). We observe mild, consistent improvements when increasing \(m\), especially on Hausdorff and DFD, while DTW remains largely stable. This suggests that hard-aware pivot learning is not overly sensitive to the exact choice of \(m\) as long as it focuses on sufficiently challenging near-neighbor pairs.

\noindent \textbf{Effect of training sample size \(|\mathcal{S}|\).} Increasing \(|\mathcal{S}|\) yields small but consistent improvements at early stages and then quickly plateaus across all distances. This indicates that the pivot learning objective can be estimated reliably using a moderate number of training pairs, with larger sample sizes achieving diminishing returns. Overall, Fig.~\ref{fig:sensitivity-geolife} supports the use of compact default settings (e.g., \(k=32\), \(m=50\), and \(|\mathcal{S}|=512\)), which achieve near-saturated ranking performance.

% Conducted on the Geolife dataset.

% Hausdorff using PIVOT method
% DTW and Frechet using PSE method

% \noindent \underline{\textbf{The effect of $k$.}} 
% 16 32 64 128 Hausdorff@Geolife DFD@Geolife

% \noindent \underline{\textbf{The effect of $m$.}} 
% 10, 20, 50, 100 Hausdorff@Geolife DFD@Geolife

% \noindent \underline{\textbf{The effect of $|\mathcal{S}|$.}} 
% 500, 1000, 2000, 3000, 4000 Hausdorff@Geolife DFD@Geolife
% 128, 256, 512, 1024 

\subsection{Efficiency}

\begin{figure*}[t]
  \centering
  \includegraphics[width=0.9\textwidth]{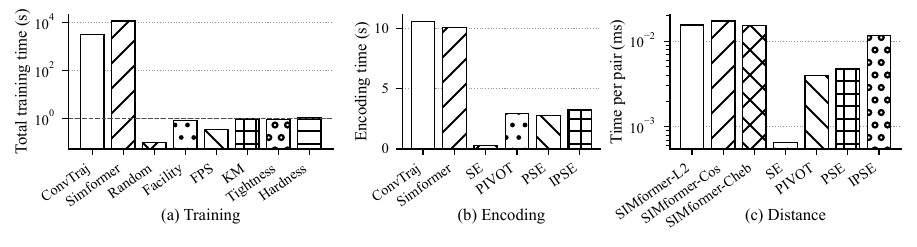}
  \vspace{-0.4cm}
  \caption{System-level efficiency: (a) training cost, (b) offline encoding time, and (c) online distance computation time per pair.}
  \label{fig:efficiency}
  \vspace{-0.3cm}
\end{figure*}

We evaluate the efficiency of the two strongest learning-based trajectory similarity methods, \textbf{ConvTraj} and \textbf{SIMformer}. Fig.~\ref{fig:efficiency} reports the training cost, offline encoding time, and online distance computation time under a unified and fair setting. All experiments are conducted on an Intel(R) Xeon(R) CPU E5-2650 v4 @ 2.20GHz with 128~GB RAM and a GeForce RTX 3080 GPU. 

For training, all learning-based methods are implemented and trained on a GPU, and both ConvTraj and SIMformer encode trajectories into 128-dimensional vectors following their original configurations. We report the total training time required to reach convergence (i.e., $200$ and $500$ training epochs for ConvTraj and SIMformer, respectively). For offline encoding, we measure the time to encode all candidate trajectories in the Porto dataset, which is also performed on a GPU to ensure consistency with the training setup. For online similarity computation, we evaluate the per-pair distance computation time by disabling parallel acceleration and randomly sampling 10{,}000 trajectory pairs, reporting the average cost per pair. Since SIMformer supports multiple vector distance functions, we report results for all distance metrics proposed in the original study. All online distance computations are conducted on a CPU, reflecting a realistic deployment scenario. 

% Overall, the results show that while learning-based methods achieve strong effectiveness, they incur substantial computational overhead during training and encoding, and their online efficiency depends heavily on the choice of representation distance, whereas the proposed lower-bound representations enable lightweight encoding and fast CPU-based similarity computation.
Overall, the results show that while the neural-based methods achieve strong effectiveness, they incur substantial computational overhead during training and encoding, and their online efficiency depends heavily on the choice of representation distance, whereas the proposed lower-bound representations enable lightweight encoding and fast CPU-based similarity computation. 
% In particular, our approach achieves up to an order-of-magnitude reduction in offline encoding time and more than an order-of-magnitude speedup in per-pair online distance computation compared with ConvTraj and SIMformer under the same 128-dimensional setting. Such significant efficiency gains translate directly into lower query latency and higher system throughput, making the proposed representations especially suitable for large-scale trajectory retrieval scenarios with millions of candidates.
In particular, 
% compared with ConvTraj and SIMformer,
% under the same 128-dimensional setting, 
\paramark{compared with ConvTraj and SIMformer, the proposed approach eliminates expensive iterative neural training and reduces offline training time by over 2,000$\times$. In addition, it reduces the offline trajectory encoding time by 3$\times$--35$\times$ and accelerates per-pair online distance computation by 3$\times$--24$\times$, depending on the representation instantiation.}
These efficiency improvements lead to lower query latency and higher system throughput, thus representing a real practical advantage of the proposed methods in large-scale settings.

\section{Related Work}
\label{sec:related_works}

% \subsection{Trajectory Similarity Learning}
\noindent \textbf{Trajectory Similarity Learning.}
The high computational cost of classical trajectory
distance measures such as DTW, Hausdorff distance, and DFD motivate trajectory similarity learning (TSL). Computing the classical distances typically involves expensive alignment or pointwise matching procedures~\cite{koide2020fast,chen2004marriage,chen2005robust,wang2018torch,yuan2019distributed,Deng2025ExactAE,deng2022efficient}, making them costly for large-scale analysis and learning.
Therefore, recent studies explore learning-based approaches that approximate trajectory similarity by mapping trajectories to fixed-length vector representations. Early sequence encoder based methods such as t2vec~\cite{li2018deep} and NeuTraj~\cite{yao2019computing} employ recurrent architectures to model trajectories as point sequences and learn embeddings that preserve similarity under different distances. ConvTraj~\cite{chang2024revisiting} revisits convolutional neural networks for trajectory similarity learning, finding that carefully designed CNN encoders, together with appropriate training objectives can achieve strong performance without relying on complex architectures. More recently, SIMformer~\cite{yang2024simformer} shows that a lightweight, single-layer Transformer encoder is sufficient for modeling free-space trajectory similarity, and also highlights the importance of tailoring representation similarity functions beyond Euclidean distance to mitigate the curse of dimensionality and improve ranking quality. Despite differences in model design, these approaches all adhere to an encoder-based paradigm, where trajectories are represented as vectors and compared using simple similarity measures.
While effective, most learning-based methods rely on complex training pipelines and distance-specific supervision, which limit their robustness and reuse across different distance measures. In contrast, the framework proposed here focuses on representation-level learning based on analytical lower bounds, providing a lightweight and distance-agnostic alternative for trajectory similarity modeling.

% \subsection{Pivot-Based Methods and Pivot Selection}
\noindent \textbf{Pivot-Based Methods and Pivot Selection.}
Pivot-based techniques have been studied extensively for accelerating similarity search in metric spaces. By precomputing distances between data objects and a small set of pivots, these methods exploit the triangle inequality to derive lower and upper bounds for pruning expensive distance computations. The effectiveness of such techniques depends heavily on the quality of the selected pivots.
Different pivot selection methods have been proposed to improve the effectiveness of pivot-based lower bounds in metric spaces. Early approaches focus on selecting well-dispersed pivots, such as random selection and farthest-first traversal (FFT)~\cite{gonzalez1985clustering}, motivated by better space coverage. Other methods exploit pivot--object distance distributions and select pivots that induce informative distance profiles, including variance- and correlation-based strategies~\cite{venkateswaran2008reference,van2011selecting}. 
% More recent approaches consider object--object distance relationships and select pivots that strengthen lower bounds over sampled object pairs, thereby improving pruning effectiveness~\cite{bustos2003pivot}. 
{\color{black}{A complementary line of work considers object--object distance relationships and selects pivots that directly strengthen lower bounds over sampled object pairs, thereby improving pruning effectiveness~\cite{bustos2003pivot}.}}
These methods are studied primarily in the context of metric indexing. We refer to a survey for additional details~\cite{zhu2022pivot}.
In contrast, our study revisits pivot selection from a representation learning perspective, optimizing pivots to construct compact lower-bound based representations that directly support ranking-oriented trajectory similarity search.

\section{Conclusion and Future Work}
\label{sec:conclusion}

% \section{Conclusion and Future Work}
We propose a pivot-based lower-bound representation framework for trajectory similarity learning. By using \emph{point}-based pivots and encoding a trajectory as a single fixed-length vector, our method provides fast similarity computation and supports multiple trajectory distances, including Hausdorff distance, DFD, and DTW, in a unified framework. We further propose tightness- and hardness-aware pivot selection strategies to improve ranking quality by learning informative pivots offline. Experiments on multiple real-world datasets show that our lower-bound representations are capable of strong and robust ranking performance, often outperforming state-of-the-art neural embeddings while remaining simple and efficient.
% There are several directions for future work. First, our framework can be extended to additional trajectory similarity measures beyond the three distances studied in this work. Second, the current representations rely only on spatial information; incorporating temporal signals (e.g., timestamps or speed profiles) may further improve ranking accuracy. Third, adapting the approach to road-network trajectories and network-constrained similarity measures is a promising direction for mobility applications.
% As future work, we plan to extend the proposed framework to more expressive trajectory similarity settings beyond free-space spatial distances. In particular, incorporating additional constraints and semantics, such as temporal information or road-network structures, into the lower-bound representation is a promising direction. Another important direction is to explore the integration of the proposed representations with indexing and filtering techniques to further improve scalability in large-scale trajectory retrieval systems.
% As future work, we plan to extend the proposed framework to more expressive trajectory similarity settings beyond free-space spatial distances. In particular, incorporating additional constraints and semantics, such as temporal information or road-network structures, into the lower-bound representation is a promising direction. Another important direction is to explore the integration of the proposed representations with indexing and filtering techniques to further improve scalability in large-scale trajectory retrieval systems.
As future work, one direction is to extend the proposed framework to more expressive trajectory similarity settings beyond free-space spatial distances, by incorporating additional constraints and semantics such as temporal information or road-networks. Another direction is to explore hybrid models that combine lower-bound representations with neural trajectory embeddings, combining the efficiency and robustness of analytical lower bounds with the representation capacity of learning-based methods, potentially achieving performance beyond either approach alone.

\clearpage
\balance
\bibliographystyle{ACM-Reference-Format}
\bibliography{sample}

\clearpage
\appendix

\section{Additional Theoretical Analysis}

\subsection{Proof of Lemma 3.2}
% \begin{lemma}[Trajectory-level lower bound induced by point pivots]
\emph{Lemma 3.2 (Trajectory-level lower bound induced by point pivots).
\label{lem:pointpivot_lb}
Let $p\in\mathbb{R}^2$ be a pivot point and define the pivot response of a trajectory $T$ as $\phi_p(T)=\min_{x\in T}\|x-p\|=\min_{x\in T} g_p(x)$.
Consider any trajectory distance $D(\cdot,\cdot)$ that \emph{dominates pointwise distances} in the sense that
for any two trajectories $T,S$:
\begin{equation}
% \label{eq:dominates_pointwise_distances}
\forall x\in T:\ \min_{y\in S}\|x-y\|\le D(T,S),
\ \ \
\forall y\in S:\ \min_{x\in T}\|x-y\|\le D(T,S)
\end{equation}
Then for any trajectories $T$ and $S$, the following lower bound holds:
\begin{equation}
% \label{eq:holds_lower_bound}
|\phi_p(T)-\phi_p(S)| \le D(T,S)
\end{equation}
% \end{lemma}
}

\begin{proof}
Fix $p\in\mathbb{R}^2$ and let $g_p(z)=\|z-p\|$ as in Lemma~\ref{lem:lipschitz}.
Let $x^\star = \arg\min_{x\in T} g_p(x)$, and $y^\star = \arg\min_{y\in S} g_p(y)$, so that $\phi_p(T)=g_p(x^\star)$ and $\phi_p(S)=g_p(y^\star)$.
We first upper bound $\phi_p(T)-\phi_p(S)$.
For any $y\in S$, by the triangle inequality, we have:
\begin{equation}
g_p(x^\star)=\|x^\star-p\| \le \|x^\star-y\|+\|y-p\|=\|x^\star-y\|+g_p(y).
\end{equation}
Rearranging terms yields $g_p(x^\star)-g_p(y)\le \|x^\star-y\|$
% \begin{equation}
% g_p(x^\star)-g_p(y)\le \|x^\star-y\|.
% \end{equation}
Taking the minimum over $y\in S$ on both sides gives
\begin{equation}
g_p(x^\star)-\min_{y\in S}g_p(y)\le \min_{y\in S}\|x^\star-y\|
\end{equation}
Using $\min_{y\in S}g_p(y)=\phi_p(S)$, we obtain:
\begin{equation}
\phi_p(T)-\phi_p(S)\le \min_{y\in S}\|x^\star-y\|
\end{equation}
By the domination assumption (applied to $x^\star\in T$), $\min_{y\in S}\|x^\star-y\|\le D(T,S)$, hence:
\begin{equation}
\phi_p(T)-\phi_p(S)\le D(T,S)
\end{equation}
% We then upper bound $\phi_p(S)-\phi_p(T)$ symmetrically.
% For any $x\in T$, the triangle inequality implies
% \[
% g_p(y^\star)=\|y^\star-p\| \le \|y^\star-x\|+\|x-p\|=\|y^\star-x\|+g_p(x),
% \]
% and thus $g_p(y^\star)-g_p(x)\le \|y^\star-x\|$.
% Taking the minimum over $x\in T$ yields
% \[
% g_p(y^\star)-\min_{x\in T}g_p(x)\le \min_{x\in T}\|y^\star-x\|.
% \]
% Using $\min_{x\in T}g_p(x)=\phi_p(T)$, we get
% \[
% \phi_p(S)-\phi_p(T)\le \min_{x\in T}\|y^\star-x\|.
% \]
% By the domination assumption (applied to $y^\star\in S$), $\min_{x\in T}\|y^\star-x\|\le D(T,S)$, hence
% \[
% \phi_p(S)-\phi_p(T)\le D(T,S).
% \]
Similarly, we can upper bound $\phi_p(S)-\phi_p(T)$ symmetrically. Combining the two inequalities gives $|\phi_p(T)-\phi_p(S)|\le D(T,S)$.
\end{proof}

\subsection{Proof of Theorem 3.3}
% \begin{lemma}[Trajectory-level lower bound induced by point pivots]
% \emph{Lemma 3.2 (Trajectory-level lower bound induced by point pivots).
% \begin{theorem}
% \label{thm:submodular}
\emph{Theorem 3.3. The tightness-aware objective function $F(P)$ is monotone and submodular. Let $P^\star$ be the optimal pivot set of size $k$, and let $P_k$ be the set returned by the greedy algorithm. Then, the greedy algorithm returns a pivot set whose objective value is guaranteed to be within a constant factor of the optimum, i.e., $F(P_k) \ge (1 - 1/e)\, F(P^\star)$.
% we have
% $F(P_k) \ge (1 - 1/e)\, F(P^\star)$.
% \end{theorem}
}
\begin{proof}[Proof Sketch]
We formalize the above intuition by showing that the objective can be decomposed into a sum of monotone submodular functions defined on individual trajectory pairs.
Fix a trajectory pair $(T,S) \in \mathcal{S}$ and define the normalized contribution of a pivot $p$ as follows:
\begin{equation}
a_p(T,S) = \frac{|\phi_p(T) - \phi_p(S)|}{D(T,S)} \ge 0
\end{equation}
The per-pair objective can be written as follows:
\begin{equation}
g_{T,S}(P) = \frac{LB_P(T,S)}{D(T,S)} = \max_{p \in P} a_p(T,S)
\end{equation}
The function $g_{T,S}(P)$ is monotone, since adding elements to $P$ cannot decrease a maximum. It is also submodular because it exhibits diminishing returns, i.e., for any $A \subseteq B$ and any $p \notin B$,
\begin{equation}
g_{T,S}(A \cup \{p\}) - g_{T,S}(A)
\ge
g_{T,S}(B \cup \{p\}) - g_{T,S}(B),
\end{equation}
as the marginal gain equals $\max(0, a_p - \max_{q\in A} a_q)$, which is no smaller than $\max(0, a_p - \max_{q\in B} a_q)$.
Since $F(P)$ is an expectation (or empirical average) of $g_{T,S}(P)$ over sampled pairs, and nonnegative linear combinations preserve monotonicity and submodularity, $F(P)$ is also monotone submodular.
Applying the classical result of Nemhauser et al.~\cite{nemhauser1978analysis} for greedy maximization of a monotone submodular function under a cardinality constraint yields the stated approximation guarantee.
\end{proof}

\section{Additional Experimental Results}

\subsection{Analysis of Tightness-Aware and Hardness-Aware Pivot Selection}

We compare tightness-aware and hardness-aware pivot selection on Geolife and Chengdu datasets under DTW, DFD, and Hausdorff distance. Following Definition~\ref{def:tightness}, pivot-selection quality is measured using the average lower-bound tightness, $\mathrm{mean}(LB_P(T,S)/D(T,S))$, where $LB_P(T,S)$ denotes the lower bound induced by the selected pivot set $P$ and $D(T,S)$ denotes the corresponding ground-truth trajectory distance.

Results are reported for pivot budgets $k\in\{16,32,64,128\}$ and training-pair sets of size $|\mathcal{S}|\in\{128,256,512,1024\}$. When varying $k$, we fix $|\mathcal{S}|=512$; when varying $|\mathcal{S}|$, we fix $k=32$. To characterize both global and retrieval-critical behavior, we evaluate the selected pivot sets on two fixed evaluation sets, namely randomly sampled trajectory pairs and hard near-neighbor pairs. The latter are constructed by selecting the nearest neighbor ($m=1$) for each anchor trajectory under the corresponding trajectory distance.

\begin{figure}[t]
    \centering
    \includegraphics[width=\columnwidth]{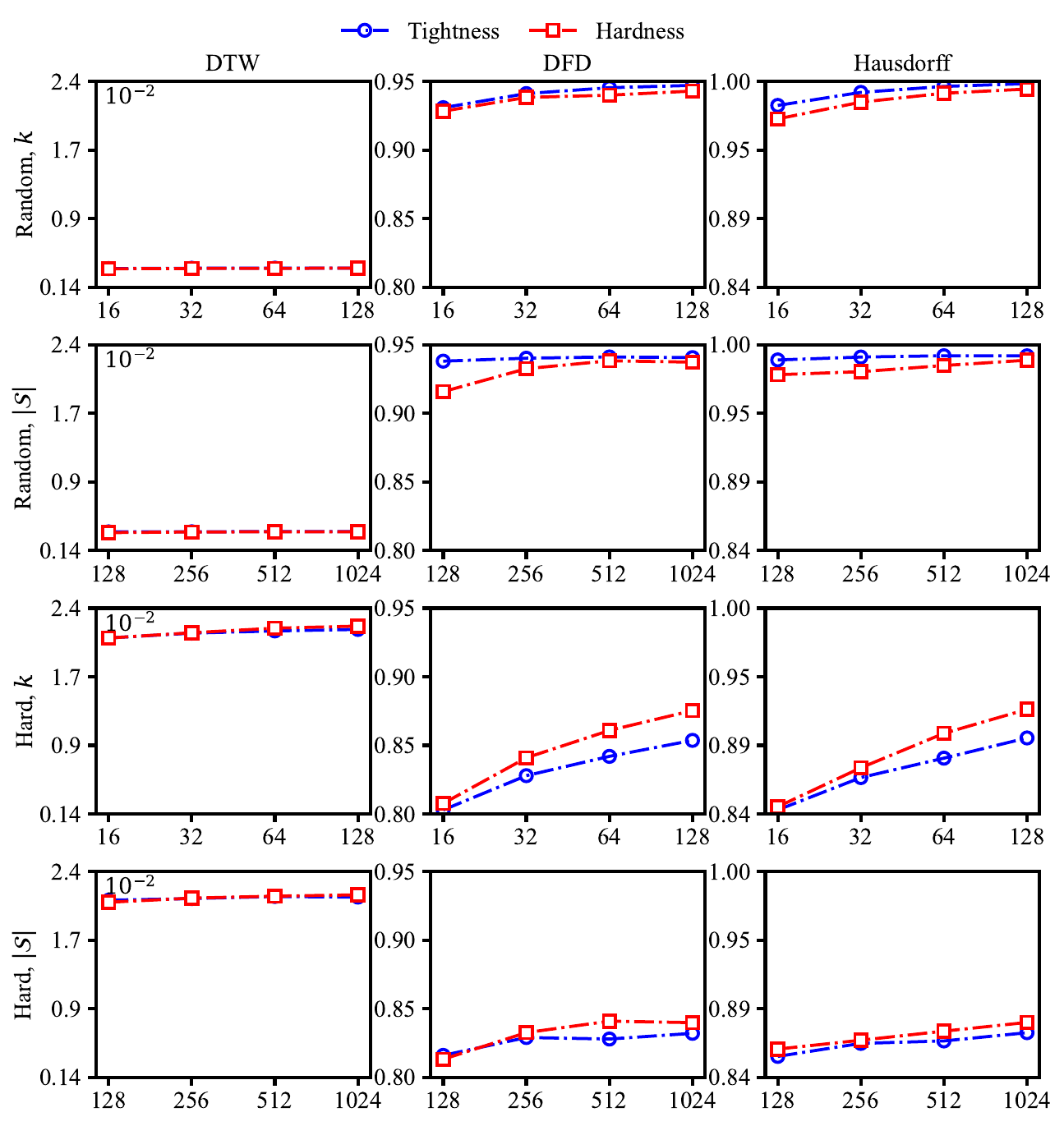}
    \caption{
    Average lower-bound tightness of tightness-aware and
    hardness-aware pivot selection on Geolife under DTW,
    DFD, and Hausdorff distance.
    }
    \label{fig:tightness_geolife}
\end{figure}

\begin{figure}[t]
    \centering
    \includegraphics[width=\columnwidth]{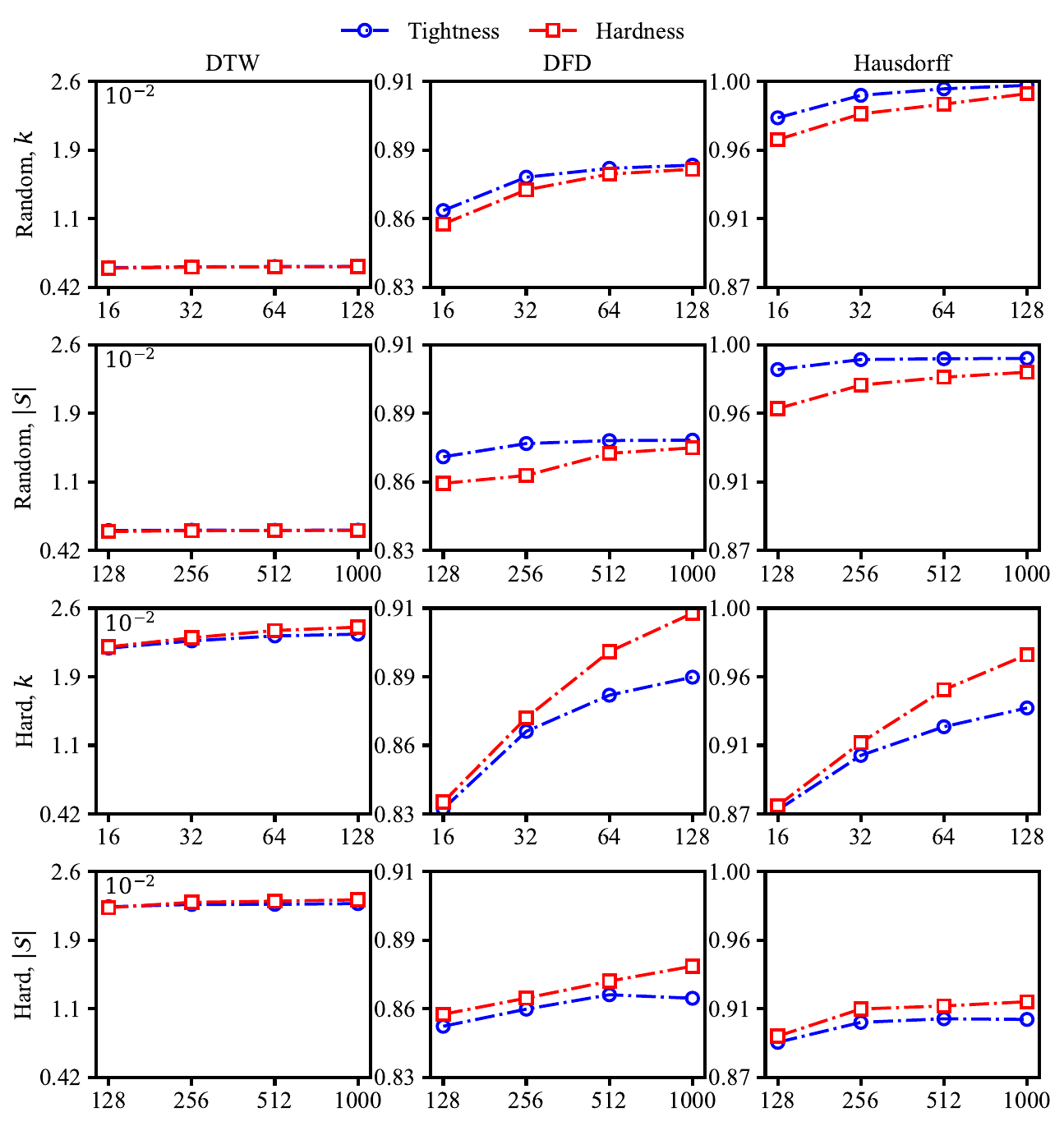}
    \caption{
    Average lower-bound tightness of tightness-aware and hardness-aware pivot selection on Chengdu under DTW, DFD, and Hausdorff distance.
    }
    \label{fig:tightness_chengdu}
\end{figure}

Figures~\ref{fig:tightness_geolife} and \ref{fig:tightness_chengdu} reveal a consistent pattern across both datasets. On randomly sampled trajectory pairs, tightness-aware selection generally achieves higher lower-bound tightness across all three distance measures, indicating stronger global lower-bound quality. In contrast, hardness-aware selection consistently produces tighter lower bounds on hard pairs, with the advantage becoming more pronounced as the pivot budget increases. 
% This behavior is expected because hardness-aware selection explicitly concentrates the optimization objective on retrieval-critical trajectory pairs rather than on uniformly sampled pairs.

These observations suggest that the two strategies optimize different aspects of representation quality. Tightness-aware selection emphasizes global lower-bound quality over the entire trajectory space, whereas hardness-aware selection prioritizes near-neighbor trajectory pairs that are more likely to influence retrieval rankings. As a result, tightness-aware selection is generally preferable when the objective is to maximize average lower-bound quality, while hardness-aware selection is more effective when the evaluation focuses on hard near-neighbor relationships.

Interestingly, higher tightness on hard pairs does not necessarily translate into better retrieval performance. For example, on the Chengdu dataset (see Table~\ref{tab:chengdu}), hardness-aware selection consistently achieves higher tightness on hard pairs, yet its HR@1 performance remains comparable to that of tightness-aware selection. This indicates that retrieval quality depends not only on lower-bound tightness, but also on whether the resulting lower-bound scores alter the relative ordering among the top-ranked candidates. Therefore, hardness-aware selection is most beneficial when improved tightness on hard pairs leads to meaningful ranking changes, whereas tightness-aware selection may remain competitive when the two strategies induce similar retrieval orderings.

\end{document}